\documentclass[letterpaper]{article} 
\usepackage{aaai23}  
\usepackage{times}  
\usepackage{helvet}  
\usepackage{courier}  
\usepackage[hyphens]{url}  
\usepackage{graphicx} 
\usepackage{natbib}  
\usepackage{caption} 
\usepackage[ruled, vlined]{algorithm2e} 

\usepackage[utf8]{inputenc} 
\usepackage[T1]{fontenc}    
\usepackage{booktabs}       

\usepackage{amsfonts}       
\usepackage{amssymb}
\usepackage{amsmath}        
\usepackage{mathtools}
\usepackage{nicefrac}       
\usepackage{microtype}      
\usepackage{amsthm}

\usepackage{xcolor}         

\usepackage{subcaption}
\usepackage[capitalise]{cleveref}

\graphicspath{ {./} }

\DeclarePairedDelimiterX{\inp}[2]{\langle}{\rangle}{#1, #2}
\newcommand{\expect}{\mathbf{E} \expectarg}
\DeclarePairedDelimiterX{\expectarg}[1]{[}{]}{%
  \ifnum\currentgrouptype=16 \else\begingroup\fi
  \activatebar#1
  \ifnum\currentgrouptype=16 \else\endgroup\fi
}

\newcommand{\innermid}{\nonscript\;\delimsize\vert\nonscript\;}
\newcommand{\activatebar}{%
  \begingroup\lccode`\~=`\|
  \lowercase{\endgroup\let~}\innermid 
  \mathcode`|=\string"8000
}

\newcommand{\EE}{\mathbf{E}}

\newcommand{\ones}{\mathbf{1}}
\newcommand{\RR}{\mathbb{R}}
\newcommand{\NN}{\mathbb{N}}
\newcommand{\dom}{{\rm dom}\,}

\newcommand{\bcdeg}{$\texttt{BCDEG}$}
\newcommand{\bcdegls}{$\texttt{BCDEG-LS}$}
\newcommand{\bcpr}{$\texttt{BCPR}$}
\newcommand{\abcpr}{$\texttt{A-BCPR}$}
\newcommand{\bcprls}{$\texttt{BCPR-LS}$}
\newcommand{\pgls}{$\texttt{PGLS}$}
\newcommand{\prls}{$\texttt{PRLS}$}
\newcommand{\pr}{$\texttt{PR}$}

\newcommand{\DKL}{D_{\rm KL}}

\DeclareMathOperator*{\argmin}{arg\,min}

\newtheorem{theorem}{Theorem}

\newtheorem{lemma}{Lemma}
\newtheorem{corollary}{Corollary}

\title{Fast and Interpretable Dynamics for Fisher Markets \\ via Block-Coordinate Updates}

\author {
    Tianlong Nan,
    Yuan Gao,
    Christian Kroer
}
\affiliations {
    Columbia University, New York, USA\\
    \{tianlong.nan, gao.yuan, christian.kroer\}@columbia.edu
}

\begin{document}
\maketitle

\begin{abstract}
    We consider the problem of large-scale Fisher market equilibrium computation through scalable first-order optimization methods. 
    It is well-known that market equilibria can be captured using structured convex programs such as the Eisenberg-Gale and Shmyrev convex programs.
    Highly performant deterministic full-gradient first-order methods have been developed for these programs. 
    In this paper, we develop new block-coordinate first-order methods for computing Fisher market equilibria, and show that these methods have interpretations as t\^atonnement-style or proportional response-style dynamics where either buyers or items show up one at a time.
    We reformulate these convex programs and solve them using proximal block coordinate descent methods, a class of methods that update only a small number of coordinates of the decision variable in each iteration.
    Leveraging recent advances in the convergence analysis of these methods and structures of the equilibrium-capturing convex programs, we establish fast convergence rates of these methods.
\end{abstract}
 
\section{Introduction} 
\label{sec:intro}

In a market equilibrium (ME) a set of items is allocated to a set of buyers via a set of prices for the items and an allocation of items to buyers such that each buyer spends their budget optimally, and all items are fully allocated. 
Due to its rich structural properties and strong fairness and efficiency guarantees,
ME has long been used to develop fair division and online resource allocation mechanisms~\citep{gao2021infinite,aziz2014cake,barman2018finding,arnsperger1994envy}.
Market model and corresponding equilibrium computation algorithms have been central research topics in market design and related areas in economics, computer science and operations research with practical impacts~\citep{scarf1967computation,kantorovich1975mathematics,othman2010finding,daskalakis2009complexity,cole2017convex,kroer2019computing}. 
More specifically, many works in market design rely on the assumption that ME can be computed efficiently for large-scale market instances. For example, the well-known fair division mechanism without money \emph{competitive equilibrium from equal incomes} (CEEI) requires computing a ME of a Fisher market under uniform buyer budgets \citep{varian1974equity}.
Recent work has also established close connections between market equilibria and important solution concepts in the context of large-scale Internet markets, such as \emph{pacing equilibria} in repeated auctions \citep{conitzer2018multiplicative,conitzer2019pacing, kroer2022market}. 
Motivated by the classical and emerging applications described above, we are interested in developing efficient equilibrium computation algorithms for large-scale market instances. 
In general, computing a ME is a hard problem~\citep{chen2009spending, vazirani2011market, othman2016complexity}. 
However, for the case of Fisher markets and certain classes of utility functions, efficient algorithms are known~\citep{devanur2008market,zhang2011proportional,gao2020first}, often based on solving a specific convex program---whose solutions are ME and vice versa---using an optimization algorithm. In this paper, we focus on two well-known such convex programs, namely, the Eisenberg-Gale (EG)~\citep{eisenberg1959consensus,eisenberg1961aggregation} and Shmyrev convex programs~\citep{shmyrev2009algorithm,cole2017convex}. 


Most existing equilibrium computation literature studies the case of a static market where all buyers and items are present in every time step of the equilibrium computation process. 
In contrast, we study a setting where only a random subset of buyer-item pairs show up at each time step. 
Such a setting is well-motivated from a computational perspective, since stochastic methods are typically more efficient for extremely large problems. Secondly, our model allows us to model new types of market dynamics.
We make use of recent advances in \emph{stochastic first-order optimization}, more specifically, block-coordinate-type methods, to design new equilibrium computation algorithms for this setting. 
The resulting equilibrium computation algorithms have strong convergence guarantees and consistently outperform deterministic full-information algorithms in numerical experiments.
In addition, many of the optimization steps not only give efficient update formulas for market iterates, but also translate to interpretable market dynamics.

\paragraph{Summary of contribution.}
We propose two stochastic algorithms for computing large-scale ME: \emph{(proximal) block-coordinate descent on EG} (\bcdeg) and \emph{block-coordinate proportional response} (\bcpr). These algorithms are derived by applying stochastic block-coordinate-type algorithms on reformulated equilibrium-capturing convex programs.
More specifically, \bcdeg\ is based on (proximal) stochastic block coordinate descent (BCD) and \bcpr\ is based on a non-Euclidean (Bregman) version of BCD.
We show that these algorithms enjoy attractive theoretical convergence guarantees, and discuss important details for efficient implementation in practice. 
Furthermore, we show that the Euclidean projection onto the simplex in \bcdeg\ (and other preexisting projected-gradient-type methods) has a t\^atonnement-style interpretation.
We then demonstrate the practical efficiency of our algorithms via extensive numerical experiments on synthetic and real market instances, where we find that our algorithms are substantially faster than existing state-of-the-art methods such as proportional response dynamics. 

\paragraph{Preliminaries and notation.}
Unless otherwise stated, we consider a \emph{linear Fisher market} with $n$ \emph{buyers} and $m$ \emph{items}. 
We use $i\in [n]$ to denote a buyer and $j\in [m]$ to denote an item.
Each buyer $i$ has a \emph{budget} $B_i > 0$ and each item has supply one.
An \emph{allocation} (or \emph{bundle}) for buyer $i$ is a vector $x_i \in \RR^m_+$ specifying how much buyer $i$ gets of each item $j$. 
Given $x_i$, buyer $i$ gets utility $\langle v_i, x_i\rangle = \sum_j v_{ij} x_{ij}$, where $v_i\in \RR^m_+$ is their valuation vector. Given \emph{prices} $p\in \RR^m_+$ (i.e., price of item $j$ is $p_j$) on all items, buyer $i$ pays $\langle p, x_i\rangle = \sum_j p_j x_{ij}$ for $x_i$.
Given prices $p$ and budget $B_i$, a bundle $x_i$ is \emph{budget feasible} for buyer $i$ if $\inp{p}{x_{i}} \le B_{i}$. 
We also use $x_{\cdot j} \in \mathbb{R}^n_+$ to denote a vector of amounts of item $j$ allocated to all buyers. 
The \emph{demand set} of buyer $i$ is the set of budget-feasible utility-maximizing allocations:
\begin{align}
    \mathcal{D}_{i}(p) = \arg\max_{x_i}\{ u_i = \inp{v_i}{x_i}: \inp{p}{x_i} \le B_i \}. \label{actual-demand-function} \tag{D}
\end{align}
A market equilibrium (ME) is an allocation-price pair $(x^*,p^*)$ such that $x^*_i \in \mathcal{D}_i(p^*)$ for all $i$ and $\sum_i x^*_{ij} \leq 1$ for all $j$, with equality if $p^*_j >0$. 

\section{Related Work}
\label{sec:relatedworks}

Since the seminal works by Eisenberg and Gale \citep{eisenberg1959consensus,eisenberg1961aggregation}, there has been an extensive literature on equilibrium computation for Fisher markets, often based on convex optimization characterizations of ME \citep{devanur2008market,zhang2011proportional,birnbaum2011distributed,cole2017convex,gao2020first,garg2006auction}. 
Equilibrium computation algorithms for more general market models with additional constraints---such as indivisible items and restrictions on the set of permissible bundles for each buyer---have also been extensively studied \citep{othman2010finding,budish2016course};
these algorithms are often based on approximation algorithms, mixed-integer programming formulations, and local search heuristics. 

In \citet{gao2020first}, the authors considered three (deterministic) first-order optimization methods, namely projected gradient (PG), Frank-Wolfe (FW) and mirror descent (MD) for solving convex programs capturing ME. 
To the best of our knowledge there are no existing results on block-coordinate methods for ME.
In the optimization literature there is an extensive and ongoing literature on new block-coordinate-type algorithms and their analysis (see, e.g., \citet{tseng2001convergence,wright2015coordinate,beck2013convergence,hanzely2019accelerated,hanzely2021fastest,liu2015asynchronous,nesterov2012efficiency,richtarik2014iteration,gao2020randomized,attouch2013convergence,zhang2020new,reddi2016fast}).
As mentioned previously, our \bcdeg\ algorithm is based on \emph{proximal block-coordinate descent} (PBCD) applied to EG.
Linear convergence of the mean-square error of the last iterate of PBCD for nonsmooth, finite-sum, composite optimization (with an objective function of the form $F(x)=\sum_i f_i(x) + \psi(x)$) has been established under different error bound conditions \citep{richtarik2014iteration,karimi2016linear,reddi2016fast}.
For \bcpr, we adopt the analysis of a recently proposed \emph{non-Euclidean} (Bregman) PBCD \citep{gao2020randomized}, which in turn made use of the convergence theory developed in \citet{bauschke2017descent}.

\section{Block Coordinate Descent Algorithm for the EG Program} 
\label{sec:eg-bcd}

(Proximal) Block coordinate descent methods (BCD) are often used to solve problems whose objective function consists of a smooth part and a (potentially) nonsmooth part. The second part is typically block-coordinate-wise separable. 
BCD algorithms update only a small block of coordinates at each iteration, which makes each iteration much cheaper than for deterministic methods, and this enables scaling to very large instances.
The EG convex program, which captures market equilibria, can be written in a form amenable to proximal block coordinate descent method.
The EG convex program for buyers with linear utility functions is
\begin{align}
    \begin{split}
        \max_x \quad & \sum\nolimits_i B_i \log{u_i} \\
        \text{s.t.} \quad & u_i \le \inp{v_i}{x_i} \quad \forall\; i \\
        & \sum\nolimits_i x_{ij} \le 1 \quad \;\forall\; j \\
        & x \ge 0. 
    \end{split}
    \tag{EG}
    \label{program:eg}
\end{align}
Any optimal solution $x^*$ to \eqref{program:eg} and the (unique) optimal Lagrange multipliers $p^*\in \RR^m_+$ associated with the constraints $\sum_i x_{ij} \leq 1$, $j \in [m]$ forms a market equilibrium. 
In fact, this holds more generally if the utility functions are concave, continuous, nonnegative, and homogeneous with degree $1$.

To apply BCD, note that \eqref{program:eg} is of the following form:
\begin{equation}
    \min_{x \in \RR^{m \times n}} F(x) := f(x) + \psi(x) = f(x) + \sum\nolimits_j \psi_{j}(x_{\cdot j})
    \label{eq:eg-bcd-formulation} 
\end{equation}
where $f(x) = - \sum\nolimits_i {B_i \log{\inp{v_i}{x_i}}}$
and $\psi_j (x_{\cdot j}) = 0$ if $\sum_i x_{ij} \le 1, x_{\cdot j} \ge 0$ and $+\infty$ otherwise. 

Thus, the nonsmooth term $\psi(x)$ decomposes along the items $j$, and we can therefore treat $x$ as a set of blocks: each item $j$ has a block of allocation variables corresponding to how much of item $j$ is given to each buyer. We use $x_{\cdot j} = (x_{1j},\ldots, x_{nj})$ to denote the $j$'th block.
Given the full gradient of $f$ at $x$ as $\nabla f(x)$, we will also need the partial gradient w.r.t. the $j$th block $x_{\cdot j}$, which we denote as $\nabla_{\cdot j} f(x)$, that is, $\nabla_{\cdot j} f(x) = \left( \nabla_{1j} f(x), \ldots, \nabla_{nj} f(x) \right)$. 


In each iteration $t$ of the BCD method, we first choose an index $j' \in [m]$ at random, with a corresponding stepsize $\eta_{j'}$. Then, the next iterate $x^+$ is generated from $x$ via $x_{\cdot j}^+ = T_j(x)$ if $j = j'$ and $x_{\cdot j}$ otherwise, 
where $T_j(x)$ equals 
\begin{equation} 
\arg\min_{y_{\cdot j}} \inp{\nabla_{\cdot j} f(x)}{y_{\cdot j} - x_{\cdot j}} + \frac{1}{\eta_{j}} \lVert y_{\cdot j} - x_{\cdot j} \rVert^2 + \psi_j(y_{\cdot j}) 
\label{eq:bcd-update}. 
\end{equation}

The above proximal mapping is equivalent to 
\begin{equation}
    T_{j}(x) = \text{Proj}_{\Delta_n} \Big( x_{\cdot j} - \eta_{j'}\nabla_{\cdot j} f(x) \Big), 
\end{equation} 
so we can generate $T_{j}(x) \in \mathbb{R}^{n}$ via a single projection onto the $n$-dimensional simplex. 
Since the projection is the most computationally expensive part, 
our method ends up being cheaper than full projected gradient descent by a factor of $m$.

To make sure $\nabla_{\cdot j} f(x)$ exists, we need to bound buyer utilities away from zero at every iteration. 
To that end, let $\underline u_i = \langle v_i, (B_i/\sum_{l=1}^n B_l) \vec{1} \rangle$ be the utility of the proportional allocation.
Then we perform ``quadratic extrapolation'' where we replace the objective function $f(x)$ with the function $\tilde{f}(x) = \sum_i \tilde{f}_i(x_i) = \sum_i \tilde{g}_i(u_i)$ where $\tilde{g}_i(u_i) = - B_i \log{u_i}$ if $u_i \ge \underline{u}_i$ and $\log{\underline{u}_i} - \frac{B_i}{\underline{u}_i} \left( u_i - \underline{u}_i \right) + \frac{B_i}{2 \underline{u}_i^2} \left(u_i - \underline{u}_i \right)^2$ otherwise. 
For \eqref{program:eg}, replacing $f$ with $\tilde{f}$ does not affect any optimal solution~\citep[Lemma 1]{gao2020first}. 

We will also need the following Lipschitz bound which ensures that the iterates are descent steps, meaning that the expected objective value is non-increasing as long as the stepsizes are not too large.
The upper bound on the allowed stepsizes (that ensure descent iterates) for a specific block is governed by the Lipschitz constant w.r.t. that block of coordinates. More details and proofs can be found in \cref{app:bcdeg}. 

\begin{lemma}
For any $j \in [m]$, let 
\begin{equation} 
    L_j = \max_{i\in [n]}{\frac{B_i v_{ij}^2}{\underline{u}_i^2}}.
\label{eq:stepsize-bcd-eg} 
\end{equation}
Then, for all $x, y \in \mathcal{X}$ such that $x, y$ differ only in the $j$th block, we have
\begin{equation}
    \tilde{f}(y) \le \tilde{f}(x) + \inp{\nabla_{\cdot j} \tilde{f}(x)}{y - x} + \frac{L_{j}}{2}{{\lVert y - x \rVert}^2}. 
    \label{eq:coordinate-continuous-3}
\end{equation}
\label{lemma:coordinate-continuous}
\end{lemma}
Let $L$ be the ``global'' Lipschitz constant which determines the stepsize of (full) gradient descent. 
Let $L_{\max} = \max_j L_j$. 
Generally, it is easy to see that $1 \le L / L_{\max} \le n$, but we can get a stronger bound than this using the gradient $(\nabla f)_{ij} = B_i v_{ij} / u_i$.
When only the variables $x_{\cdot j}$ in block $j$ change, we have that $L_j$ is bounded by the maximal diagonal value of the $j$'th-block (sub) Hessian matrix.
In contrast, for $L$ it depends on the maximal trace of the $i$'th-block (sub) Hessian matrix over all buyers $i$.
From a market perspective this can be interpreted as follows: when we only adjust one block, each buyer's utility fluctuates based on that one item. On the contrary, for full gradient methods, every item contributes to the change of $u_i$.
This yields that the ratio of $L_j / L$ is roughly $\max_i v_{ij}^2 / \max_i \lVert v_i \rVert^2$.

\begin{algorithm}
\KwIn{Initial $x^0$, stepsizes $\eta_1^k, \eta_2^k, \ldots, \eta_m^k,\; \forall\, k \in \NN$}
	\For{$k \gets 1, 2, \ldots $} {
		pick $j_k \in [m]$ with probability $1 / m$\;
		$g^{k-1} \gets \nabla_{\cdot j_k} \tilde{f}(x^{k-1})$\;
		$x^k_{\cdot j_k} \gets \text{Prox}_{\Delta_n} \left( x^{k-1}_{\cdot j_k} - \eta_{j_k}^k g^{k-1} \right)$\; 
        $x^{k}_{\cdot j} \gets x^{k-1}_{\cdot j},\; \forall\, j \neq j_k$\;
	}
\caption{\bcdeg\: (Proximal) Block Coordinate Descent for the EG Program}
\label{algo:bcdeg}
\end{algorithm}

\cref{algo:bcdeg} states the (proximal) BCD algorithm for the EG convex program.
Each iteration only requires a single projection onto an $n$-dimensional simplex, as opposed to $m$ projections for the full projected gradient method. 
Moreover, we show in \cref{app:bcdeg} that for linear utilities we can further reduce the computational cost per iteration of \cref{algo:bcdeg} when the valuation matrix $v_{ij}$ is sparse. 

\paragraph{Line search strategy.}
As mentioned in previous work such as \citet{richtarik2014iteration}, line search is often very helpful for BCD. 
If it can be performed cheaply, then it can greatly speed up numerical convergence. We show later that this occurs for our setting.

We incorporate line search in \cref{algo:bcdeg}  with a (coordinate-wise) $l_2$ smoothness condition.
The line search modifies \cref{algo:bcdeg} as follows: after computing $x_{\cdot j_k}^k$, we  check whether $$\eta_{j_k} \lVert \nabla_{\cdot j_k} \tilde{f}(x^k) - \nabla_{\cdot j_k} \tilde{f}(x^{k-1}) \rVert \leq \lVert x^k_{\cdot j_k} - x^{k-1}_{\cdot j_k} \rVert.$$ If this check succeeds then we increase the stepsize by a small multiplicative factor and go to the next iteration. If it fails then we decrease by a small multiplicative factor and redo the calculation of $x_{\cdot j_k}^k$.
This line search algorithm can be implemented in $O(n)$ cost per iteration, whereas a full gradient method requires $O(nm)$ time. 
A full specification of \bcdeg\ with line search (\bcdegls) is given in~\cref{app:bcdeg}.

\paragraph{Convergence analysis.}
Next we establish the linear convergence of \cref{algo:bcdeg} under reasonably-large (fixed) stepsizes, as well as for the line search variant. 

Following prior literature on the linear convergence of first-order methods for structured convex optimization problems under ``relaxed strong convexity'' conditions, we show that \bcdeg\
generates iterates that have linear convergence of the expected objective value.
For more details on these relaxed sufficient conditions that ensure linear convergence of first-order methods, see \citet{karimi2016linear} and references therein. 

\citet{gao2020first} showed that \eqref{program:eg} and other equilibrium-capturing convex programs can be reformulated to satisfy these conditions. Hence, running first-order methods on these convex programs yields linearly-convergent equilibrium computation algorithms. 
Similar to the proof of \citet[Theorem 2]{gao2020first} and \citet[(39)]{karimi2016linear}, we first establish a \emph{Proximal-P\L}\ inequality. 

\begin{lemma}
    For any feasible $x$ and any $L > 0$, define $\mathcal{D}_{\psi}(x, L)$ to be
    \begin{equation} 
        -2 L \min_y \inp{\nabla \tilde{f}(x)}{y - x} + \frac{L}{2} {\lVert y - x \rVert}^2 + \psi(y) - \psi(x)
        \nonumber
    \end{equation}
    then we have the inequality
    \begin{equation}
        \frac{1}{2}\mathcal{D}_{\psi}(x, L) \ge \min\left\{\frac{\mu}{\theta^2(A, C)}, L\right\} (F(x) - F^*)
        \label{eq:proximal-pl-condition}
    \end{equation}
    where $\theta(A, C)$ is the Hoffman constant of the polyhedral set of optimal solutions of \eqref{program:eg}, which is characterized by matrices $A$ and $C$ where $C$ is a matrix capturing optimality conditions,  and $F^*$ is the optimal value.
    \label{lemma:proximal-pl}
\end{lemma}

In previous inequalities, they essentially showed that for any $L \ge \lambda = {\mu}/{\theta^2(A, C)} > 0$, 
$\frac{1}{2}\mathcal{D}_\psi(x, L) \ge \lambda (F(x) - F^*)$.
However, here we generate a Proximal-P\L\ inequality giving a lower bound on $\frac{1}{2}\mathcal{D}_{\psi}(x, L)$ for any $L > 0$. 
Then, combining \cref{lemma:proximal-pl} with \citet[Lemmas 2 \& 3]{richtarik2014iteration}, we can establish the main convergence theorem for \bcdeg. 
\begin{theorem}
    Given an initial iterate $x^0$ and stepsizes $\eta_j^0 = 1/L_j,\, \forall\; j$ satisfying \eqref{eq:stepsize-bcd-eg}, let $x^k$ be the random iterates generated by \cref{algo:bcdeg}.
    Then, 
\begin{equation}
    \EE\big[F(x^{k+1})\big] - F^* \le (1 - \rho)^k \left( F(x^0) - F^* \right),
\label{eq:linear-convergence-bcd-eg} 
\end{equation}
where $\rho = \min \left\{\frac{\mu}{m L_{\max} \theta^2(A, C)}, \frac{1}{m}\right\}$ and $L_{\max} = \max_j L_j$. 
\label{theorem:bcdeg-convergence}
\end{theorem}
\citet{karimi2016linear} also develop a block-coordinate method that applies to EG.
Unlike their result, our result admits larger stepsizes that can vary per block, as well as a line search strategy, which is helpful for practical performance as we show in \cref{sec:experiments}.

\section{Economic Interpretation of Projected Gradient Steps}

As \citet{goktas2021consumer} argue, one drawback of computing market equilibrium via projected gradient methods such as \bcdeg\ is that these methods do not give a natural interpretation as market dynamics. To address this deficiency, in this section we show that \cref{algo:bcdeg}, and projected gradient descent more generally, can be interpreted as distributed pricing dynamics that balance supply and demand. 

The projection step in \cref{algo:bcdeg}, for an individual buyer $i$ and a chosen item $j$, is as follows (where we use $j$ for $j_k$ and drop the time index for brevity): 
\begin{equation}
    x^k_{\cdot j} \gets \text{Proj}_{\Delta^n} \left(x^{k-1}_{\cdot j} - \eta_j g^{k-1}\right).
\end{equation} 

As is well-known, the projection of a vector $y\in \RR^n$ onto the simplex $\Delta^n = \left\{ x\in \RR^n_+: \sum_i x_i = 1  \right\}$ can be found using an $O(n\log n)$ algorithm (the earliest discovery that we know of is \citet{held1974validation}; see \cref{app:econ-projection} for a discussion of more recent work on simplex projection). 
The key step is to find the (unique) number $t$ such that 
\begin{align}
    \sum\nolimits_i (y_i - t)^+ = 1 
    \label{eq:simplex-find-t}
\end{align}
 and compute the solution as $x = (y - t\cdot \ones)^+$ (component-wise). This can be done with a simple one-pass algorithm if the $y_i$ are sorted. 
 In fact, the number $t$ corresponds to the (unique) optimal Lagrange multiplier of the constraint $\sum_i x_i = 1$ in the KKT conditions. 

In the projection step in \cref{algo:bcdeg}, \eqref{eq:simplex-find-t} has the form
\begin{align}
    \sum\nolimits_i ( x_{ij}^{k-1} - \eta_j g^{k-1}_i - t )^+ = 1.
    \label{eq:bcdeg-simplex-find-t}
\end{align}
Recall that 
we have $g^{k-1} < 0$. Note that the left-hand side of \eqref{eq:bcdeg-simplex-find-t} is non-increasing in $t$ and strictly decreasing around the solution (since some terms on the left must be positive for the sum to be $1$). Furthermore, setting $t=0$ gives a lower bound of $\sum_i ( x_{ij}^{k-1} - \eta_j g^{k-1}_i - t )^+ > \sum_i x^{k-1}_{ij} = 1$. Hence, the unique solution $t^*$ must be positive. 
Now we rewrite $t$ as $\eta_j p_j$ for some ``price'' $p_j$. Then, \eqref{eq:bcdeg-simplex-find-t} can be written as
\begin{equation}
    \sum\nolimits_i D^k_i(p_j) = 1,\;
    D^k_i(p_j) = ( x^{k-1}_{ij} - \eta_j g^{k-1}_i - \eta_j p_j )^+.
    \label{eq:bcdeg-find-pj}
\end{equation}
In other words, the projection step is equivalent to finding $p_j^k$ that solves \eqref{eq:bcdeg-find-pj}. 
Here, $D^k_i$ can be viewed as the \emph{linear demand function} of buyer $i$ at time $k$ given a prior allocation $x^{k-1}_{ij}$.
The solution $p_j^k$ can be seen as a \emph{market-clearing price}, since $\sum_i D^k_i(p_j^k)$ exactly equals the unit supply of item $j$. 
After setting the price, the updated allocations can be computed easily just as in the simplex projection algorithm, that is, $x^k_{ij} = \left( x^{k-1}_{ij} - \eta_j g^{k-1} - \eta_j p^k_j \right)^+$ for all $i$. Equivalently, the new allocations are given by the current linear demand function: $x^k_{ij} = D^k_i(p^k_j)$.

Summarizing the above, we can recast \cref{algo:bcdeg} into the following dynamic pricing steps. At each time $k$, the following events occur. 
\begin{itemize}
    \item An item $j$ is sampled uniformly at random. 
    \item For each buyer $i$, her demand function becomes $D_i^k(p_j) \gets ( x^{k-1}_{ij} - \eta_j g^{k-1} - \eta_j p_j )^+$. 
    \item Find the unique price $p_j^k$ such that $\sum_i D^k_i(p^k_j)=1$. 
    \item Each buyer chooses their new allocation of item $j$ via $x^k_{ij} = D^k_i(p^k_j)$. 
\end{itemize}

To get some intuition for the linear demand function, note that when $u_i \ge \underline{u}_i$, we have $g^{k-1}_i = - {B_i v_{ij}}/{u^{k-1}_i}$, and therefore it holds that $D^k_i\left( {B_i v_{ij}}/{ u^{k-1}_i } \right) = x^{k-1}_{ij}$.
In other words, the current demand of buyer $i$ is exactly the previous-round allocation $x^{k-1}_{ij}$ if the price of item $j$ is $(B_i / u^{k-1}_i)v_{ij}$. 
This can be interpreted in terms familiar from the solution of EG: let $\beta^{k-1}_i = B_i / u^{k-1}_i$ be the \emph{utility price} at time $k-1$ for buyer $i$, then we get that after seeing prices $p_j^{k}$, buyer $i$ increases their allocation of goods that beat their current utility price, and decreases their allocation on goods that are worse than their current utility price. The stepsize $\eta_j$ denotes buyer $i$'s responsiveness to price changes. 

The fact that the prices are set in a way that equates the supply and demand is reminiscent of t\^atonnement-style price setting. The difference here is that the buyers are the ones who slowly adapt to the changing environment, while the prices are set in order to achieve exact market clearing under the current buyer demand functions.

\section{Relative Block Coordinate Descent Algorithm for PR dynamics}
\label{sec:rcd}

\citet{birnbaum2011distributed} showed that the \emph{Proportional Response dynamics} (\pr) for linear buyer utilities can be derived by applying mirror descent (MD) with the KL divergence on the Shmyrev convex program formulated in buyers' \emph{bids}~\citep{shmyrev2009algorithm, cole2017convex}. 
The authors derived an $O(1/k)$ last-iterate convergence guarantee of \pr\ (MD) by exploiting a ``relative smoothness'' condition of the Shmyrev convex program. 
This has later been generalized and led to faster MD-type algorithms for more general relatively smooth problems \citep{hanzely2021fastest,lu2018relatively,gao2020randomized}.
In this section, we propose a randomized extension of \pr\ dynamics, which we call \emph{block coordinate proportional response} (\bcpr). \bcpr\ is based on a recent stochastic mirror descent algorithm \cite{gao2020randomized}. We provide stepsize rules and show that each iteration involves only a single buyer and can be performed in $O(m)$ time. 

Let $b\in \RR^{n\times m}$ denote the matrix of all buyers' bids on all items and $p_j(b):=\sum_i b_{ij}$ denote the price of item $j$ given bids $b$. 
Denote $a_{ij} = \log v_{ij}$ if $v_{ij} > 0$ and $0$ otherwise. 
The Shmyrev convex program is
\begin{align}
    \begin{split}
    \max_{b} \quad & \sum\nolimits_{i,j} a_{ij} b_{ij} - \sum\nolimits_j p_j(b) \log p_j(b) \\
    \text{s.t.} \quad & \sum\nolimits_j b_{ij} = B_i \quad \forall\, i \\ 
    & b \ge 0.  
    \end{split} 
    \tag{S} 
    \label{program:shmyrev}
\end{align}
It is known that an optimal solution $(b^*, p^*)$ of \eqref{program:shmyrev} gives equilibrium prices $p^*_j$. Corresponding equilibrium allocations can be constructed via $x_{ij}^* = b_{ij}^* / p_j^*$ for all $i,j$. 
\eqref{program:shmyrev} can be rewritten as minimization of a smooth finite-sum convex function with a potentially nonsmooth convex separable regularizer:
\begin{align}
    \min_{b \in \RR^{m \times n}} \Phi(b) := \varphi(b) + r(b) = \varphi(b) + \sum\nolimits_i r_i(b_i) \label{eq:rcd-formulation}
\end{align}
where $\varphi(b) = - \sum\nolimits_{i, j} {b_{ij} \log\left(\frac{v_{ij}}{p_j(b)}\right)}$ 
and $r_i(b_i) = 0$ if $\sum_j b_{ij} = B_i,\, b_i \ge 0$ and $+\infty$ otherwise. 

Now we introduce the relative randomized block coordinate descent (RBCD) method for \eqref{eq:rcd-formulation}. 
We use the KL divergence as the Bregman distance in the proximal update of $b$. 
Let $\DKL(q^1, q^2) = \sum_{i}{q^1_i \log{(q^1_i / q^2_i)}}$ denote the KL divergence between $q^1$ and $q^2$ (assuming $\sum_i q^1_i = \sum_i q^2_i$ and $q^1_i, q^2_i > 0$, $\forall\, i$). 
In each iteration, given a current $b$, we select $i \in [n]$ uniformly at random and only the update $i$-th block of coordinates $b_{i}$. The next iterate $b^+$ is $b_i^+ = T_i(b_i)$ for $i$ and $b^+ = b$ for the remaining blocks. Here, $T_i(b_i)$ equals 
\begin{equation}
    \arg\min_a \inp{\nabla_i \varphi(b)}{a - b_i} + \frac{1}{\alpha_i} \DKL(a, b_i) + r_i(a)
    \label{eq:bcpr-update}
\end{equation}
where $\alpha_i > 0$ is the stepsize. 
It is well-known that \eqref{eq:bcpr-update} is equivalent to 
the following simple, explicit update formula:  
\begin{equation}
    b_{ij}^+ = \frac{1}{Z_i} b_{ij} \left( \frac{v_{ij}}{p_j} \right)^{\alpha_i} \quad \forall j = 1, 2, \ldots, m 
    \label{eq:bcpr-update-simple}
\end{equation}
where $Z_i$ is a normalization constant such that $\sum_j b_{ij}^+ = B_i$. 

\cref{algo:bcpr} states the full block coordinate proportional response dynamics.
\begin{algorithm}
    \KwIn{Initial $b^0$, $p^0$, stepsizes $\alpha_1^k, \ldots, \alpha_n^k,\; \forall\, k \in \NN$}
        \For{$k \gets 1, 2, \ldots $} {
            pick $i_k \in [n]$ with probability $1/n$\; 
            compute $b_{i_k}^+$ based on \eqref{eq:bcpr-update-simple} with stepsize $\alpha_{i_k}^k$\;
            $b_{i_k}^k \gets b_{i_k}^+$ and $b_{i'}^k \gets b_{i'}^{k-1}, \; \forall\, i' \neq i_{k}$\;
            $p^k_j \gets \sum_i b_{ij}^k, \; \forall\, j$\;
        }
\caption{Block Coordinate Proportional Response (BCPR)}
\label{algo:bcpr}
\end{algorithm}




\paragraph{Convergence Analysis.}
The objective function of \eqref{program:shmyrev} is relatively smooth with $L = 1$~\citep[Lemma 7]{birnbaum2011distributed}. This means that $\alpha_i = 1$ is a safe lower bound for stepsizes. 
For \cref{algo:bcpr}, a last-iterate sublinear convergence rate is given by \citet{gao2020randomized} for $0 < \alpha_i < \frac{1 + \theta_i}{L_i}$, where $\theta_i$ is the Bregman symmetry measure introduced by \citet{bauschke2017descent}. 
For the KL divergence $\theta_i = 0$.
Their proof still goes through for $\alpha_i = 1 / L_i$, which yields the following
\begin{theorem} 
    Let $b^k$ be random iterates generated by \cref{algo:bcpr} with $\alpha_i^k = 1 / L_i$ for all $k$, then 
    \begin{equation}
        \EE \big[ \Phi(b^k) \big] - \Phi^* \le \frac{n}{n + k} \Big( \Phi^* - \Phi(b^0) + D(b^*, b^0) \Big)
    \end{equation}
    where $\Phi^*$ is the optimal objective value. 
    \label{thm:bcpr-fixed-stepsize}
\end{theorem}

\paragraph{Line search.} 
To speed up \bcpr, we introduce a line search strategy and an adaptive stepsize strategy. 
\bcpr\ with line search can be implemented by comparing $\DKL(p^+, p)$ and $\DKL(b_i^+, b_i)$, which takes $O(m)$ time, and is much cheaper than computing the whole objective function value of \eqref{program:shmyrev}. 
Beyond that, by storing $p$, we also avoid touching all variables in each iteration. Therefore, the amount of data accessed and computation needed is $O(m)$ per iteration (vs. $O(nm)$ for full gradient methods). 
In the adaptive strategy we compute larger stepsize based on Lipschitz estimates using a closed-form formula. 
The \bcpr\ with Line Search (\bcprls) and Adaptive \bcpr\ (\abcpr) are formally stated in \cref{app:rbcd}.  

As with \bcdeg, our experiments demonstrate that larger stepsizes can accelerate \cref{algo:bcpr}. When we consider a series of stepsizes $\{ \alpha_i^k \}_{k \in \NN}$ generated by line search or an adaptive strategy, we can show the inequality 
\begin{align}
    &\expect*{ \Phi(b^k) - \Phi(b^*) } \nonumber \\ 
    & \le \frac{n}{n + k} \left( \Phi(b^*) - \Phi(b^0) + \sum_i \frac{1}{\alpha_i^0} D(b_i^*, b^0) \right) \nonumber \\ 
    & + \frac{1}{k} \sum_{l=1}^k \expect*{ \sum_i \frac{1}{\alpha_i^l} D(b_i^*, b^l) - \sum_i \frac{1}{\alpha_i^{l-1}} D(b_i^*, b^l) }. 
    \label{eq:bcpr-ls-bound}
\end{align} 
However, we cannot guarantee convergence, as we are unable to ensure the convergence of the last term above. 

\section{Proportional Response with Line Search}
In this section we extend vanilla PR dynamics to PR dynamics with line search (\prls), by developing a Mirror Descent with Line Search (MDLS) algorithm. 
Intuitively, the LS strategy is based on repeatedly incrementing the stepsize and checking the relative-smoothness condition, with decrements made when the condition fails.
This is similar to the projected gradient method with line search (\pgls) in \citet[A.5]{gao2020first}, but replaces the $\ell_2$ norm with the Bregman divergence. 
This allows larger stepsizes while guaranteeing the same sublinear last-iterate convergence. 
The general  MDLS algorithm is stated in \cref{app:prls}. 

\paragraph{Convergence rate.} \citet[Theorem 3]{birnbaum2011distributed} showed that a constant stepsize of $\alpha^k = 1/L$ ensures sublinear convergence at a $1/k$ rate.
One of the key steps in establishing the rate is the descent lemma, which also holds in the line search case:
\begin{lemma}
    Let $b^k$ be MDLS iterates. 
    Then, $D(b^*, b^{k+1}) \le D(b^*, b^k)$ for all $k$.
    \label{lem:mdls-d-descent}
\end{lemma}

We have the following theorem. 
\begin{theorem}
    Let $b^k$ be iterates generated by \prls\
    starting from any initial solution feasible $b^0$, then we have
    \begin{equation} 
        \varphi(b^k) - \varphi^* \leq \frac{1}{\rho^-} \cdot \frac{ D(b^*, b^0) }{k} 
        \label{eq:mdls-non-tighter-convergence}
    \end{equation}
    where $\rho^-$ is the shrinking factor of stepsize. 
    \label{thm:mdls-sublinear-conv}
\end{theorem}
Unlike the result in \bcprls, the deterministic algorithm maintains its convergence guarantee with line search. The main issue in the block coordinate case is the lack of monotonicity of $\EE[D_i(b^*, b^k)]$, which is avoided by the deterministic algorithm. In the proof, we also give a tighter, path-dependent bound.



\section{Block Coordinate Descent Algorithm for CES utility function}

In this section, we show that block coordinate descent algorithms also work for the case where buyers have \emph{constant elasticity of substitution} (CES) utility functions (for $\rho \in (0, 1)$). 

Formally, a CES utility function for buyer $i$ with parameters $v_i \in \RR^m_{+}$ and $\rho \in (0, 1)$ is $u_{i}(x_i) = ( \sum_j {v_{ij} x_{ij}^\rho} )^{\frac{1}{\rho}}$ where $v_{ij}$ denotes the valuation per unit of item $j$ for buyer $i$. 
CES utility functions are convex, continuous, nonnegative and homogeneous and hence the resulting ME can still be captured by the EG convex program. 

\paragraph{BCDEG for CES utility function.}
The resulting EG program is of similar form as \eqref{eq:eg-bcd-formulation} with
    $f(x) = - \sum\nolimits_{i} {\frac{B_i}{\rho} \log{\inp{v_{i}}{x_i^\rho}}}$,
where $x_i^\rho = (x_{i1}^\rho, \ldots, x_{im}^\rho)$ and the same separable $\psi(x)$ as \eqref{eq:eg-bcd-formulation}. Hence, as for linear Fisher markets, we can apply block-coordinate descent. 


Since $u$ and $x$ may reach $0$ at some iterates, we need smooth extrapolation techniques to ensure the existence of gradients. 
First, similar to \citet[Lemma 8]{zhang2011proportional}, we lower bound $x^*$ for all $i, j: v_{ij} > 0$, which ensures that extrapolation will not affect the equilibrium when $\min_{i, j} v_{ij} > 0$. 
Our bounds are tighter than \citet{zhang2011proportional}. 
\begin{lemma}
    For a market with CES utility functions with $\rho \in (0, 1)$ and market equilibrium allocation $x^*$, for any $i, j$ such that $v_{ij} > 0$, we have $x_{ij}^* \ge \underline{x}_{ij}^* := \omega_1(i)^{\frac{1}{1-\rho}} \omega_2(i)^{\frac{\rho(\rho+1)}{1 - \rho}}$, where $\omega_1(i) = \frac{B_i}{m\sum_l B_l}$, $\omega_2(i) = \min_{j: v_{ij} > 0} v_{ij} / \max_j v_{ij}$.
    \label{lem:ces-helper-lemma}
\end{lemma}
In \cref{app:ces}, we show how to use this bound to perform safe extrapolation. This yields the following theorem for applying \bcdeg\ to CES utilities. Due to its similarity to the linear case, we give the full algorithm in the appendix. 
The theorem is a direct consequence of \cref{lem:ces-helper-lemma} and \citet[Theorem 7]{richtarik2014iteration}.
\begin{theorem}
    Let $x^k$ be the random iterates generated by \bcdeg\; for CES utility function ($\rho \in (0, 1)$) with stepsizes $\eta_j^k = 1 / L_j \; \forall\; k$, where
    \begin{equation} 
        L_j = \max_{i\in [n]}{\frac{B_i v_{ij} \underline{x}_{ij}^{\rho - 2}}{\underline{u}_i(\rho)}} \quad \text{for all } j \in [m] 
        \label{eq:stepsize-bcdeg-ces} 
    \end{equation}
    and $\underline{u}_i(\rho) = \sum_j v_{ij} {\underline{x}^*_{ij}}^\rho$ (defined in \cref{lem:ces-helper-lemma}).
    Then,  
    \begin{equation}
        \EE\big[F(x^k)\big] - F^* \le \Big( 1 - \frac{\mu(L)}{m} \Big)^k \left( F(x^0) - F^* \right)
        \label{eq:ces-coordinate-continuous-3}
    \end{equation}
    where $\mu(L)$ is the strong-convexity modulus w.r.t. the weighted norm $\sum_j L_j \lVert \cdot \rVert^2_{\cdot j}$. 
    \label{theorem:ces-convergence}
\end{theorem}


\paragraph{BCPR for CES utility function.}
Unlike for linear utilities, the EG program for CES utility cannot be converted to a simple dual problem. Hence, we cannot view PR for $\rho \in (0, 1)$ as a mirror-descent algorithm and analyze it with typical relative smoothness techniques. However, \citet{zhang2011proportional} nonetheless showed convergence of PR for $\rho \in (0, 1)$. We show that we can still extend their proof to show convergence of block coordinate PR for CES utility. 
\begin{theorem}
    Let $b^k$ be the random iterates generated by \bcpr\; for CES utility function ($\rho \in (0, 1)$). For any $\epsilon > 0$, when 
    \begin{equation}
        k \ge \frac{ 2 \log{\frac{\sqrt{8 D(b^*, b^0)} W^{\frac{1}{1-\rho}}}{\epsilon}} }{\log{\frac{n}{n - 1 + \rho}}},
        W = \frac{n}{\min_{v_{ij} > 0} v_{ij} \cdot \min_i B_i}, 
    \end{equation}
    we have $\expect*{ \frac{\lvert b_{ij} - b_{ij}^* \rvert}{b_{ij}^*} } \le \epsilon$ for all $i, j$ such that $v_{ij} > 0$. 
    \label{thm:bcpr-ces-convergence}
\end{theorem}

\section{Numerical Experiments} 
\label{sec:experiments}

We performed numerical experiments based on both simulated (for linear and CES ($\rho \in (0, 1)$) utilities) and real data to test the scalability of our algorithms.

To measure the amount of work performed, we measure the number of accesses to cells in the valuation matrix.
For the deterministic algorithms, each iteration costs $n \times m$, whereas for our \bcdeg\ algorithms each iteration costs $n$, and for our \bcpr\ algorithms each iteration costs $m$. For algorithms that employ line search, we count the valuation accesses required in order to perform the line search as well.

To measure the accuracy of a solution, we use the duality gap and average relative difference between $u$ and $u^*$.
The instances are small enough that we can compute the equilibrium utilities $u^*$ using \citet{mosek2010mosek}.

\paragraph{Simulated low-rank instances.}
To simulate market instances, we generate a set of valuations that mimic \emph{approximately low rank valuations}, which are prevalent in real markets, and were previously studied in the market equilibrium context by \citet{kroer2019computing}.
The valuation for item $j$ and buyer $i$ is generated as: 
    $v_{ij} = v_i v_j + \epsilon_{ij}$, where $v_i \sim \mathcal{N}(1, 1)$, $v_j \sim \mathcal{N}(1, 1)$, and $\epsilon_{ij} \sim \emph{uniform}(0, 1)$.
Here, buyer $i$'s valuation for item $j$ consists of three parts: a value of item $j$ itself ($v_j$), buyers $i$'s average valuation ($v_i$), and a random term $\epsilon_{ij}$. 
We consider markets with $n = m = 400$. 
All budgets and supplies are equal to one. 

\paragraph{Movierating instances.}
We generate a market instance using a movie rating dataset collected from twitter called \emph{Movietweetings}~\citep{dooms2013movietweetings}. 
Here, users are viewed as the buyers, movies as items, and ratings as valuations. 
Each buyer is assigned a unit budget and each item unit supply. We use the ``snapshots $200$K'' data set and remove users and movies with too few entries. Using the matrix completion software \emph{fancyimpute}~\citep{fancyimpute}, we estimate missing valuations. 
The resulting instance has $n = 691$ buyers and $m = 632$ items. 

First we compare each of our new algorithms in terms of the different stepsize strategies: \bcdeg\ vs. \bcdegls, \bcpr\ vs. \abcpr\ vs. \bcprls, and \pr\ vs \prls.
The results are shown in \cref{fig:compare-simulated-data} and \cref{fig:compare-real-data} in the upper left, upper right, and lower left corners. 
In all cases we see that our new line search variants perform the best.

Second, we then compare our block-coordinate algorithms to the best deterministic state-of-the-art market equilibrium algorithms: \pgls\ \citep{gao2020first} and \prls.
The results are shown in \cref{fig:compare-simulated-data} and \cref{fig:compare-real-data} in the lower right corner.
\begin{figure}
        \centering
        \includegraphics[width=0.23\textwidth]{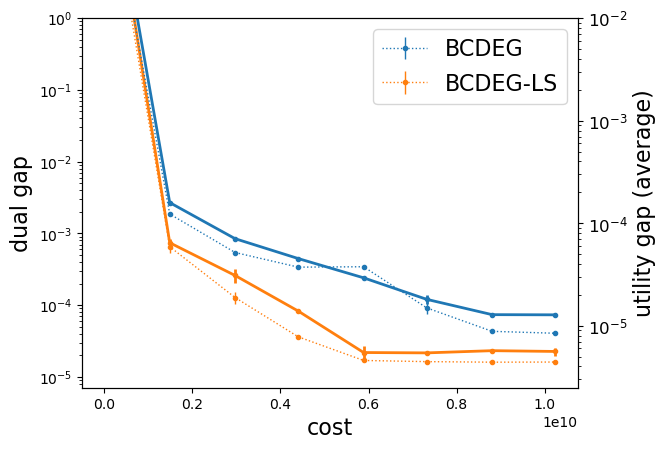}
        \includegraphics[width=0.23\textwidth]{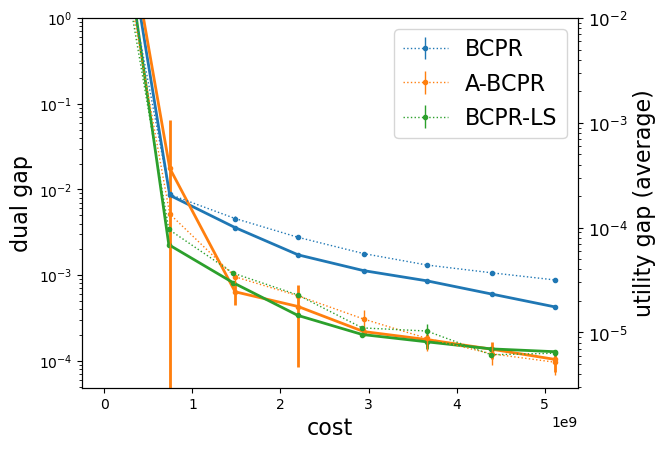}
        \includegraphics[width=0.23\textwidth]{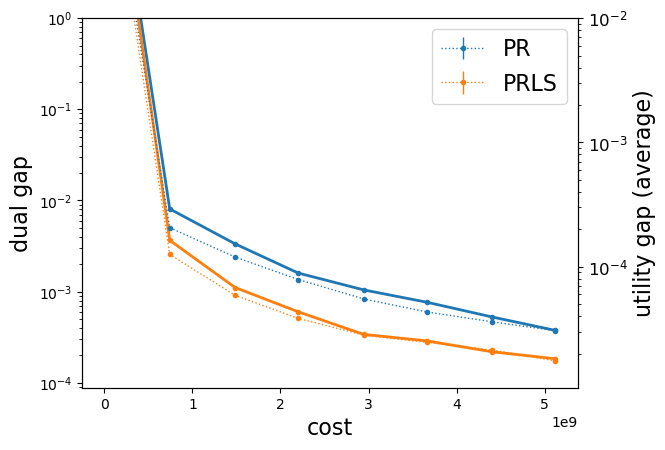}
        \includegraphics[width=0.23\textwidth]{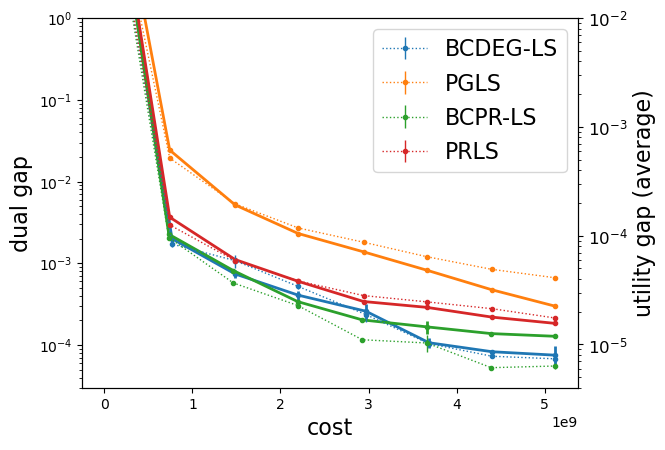}
    \caption{Performance on simulated low-rank instances. Random algorithms were implemented with seeds $0-9$. We also plotted vertical bars representing standard deviations across different seeds. The left y-axis shows performance in terms of the duality gap (solid lines for each algorithm) while the right y-axis shows performance in terms of utilities (dotted line for each algorithm). The x-axis shows units of work performed. }
    \label{fig:compare-simulated-data}
\end{figure}
For all markets either \bcdeg\ or \bcprls\ is best on all  metrics, followed by \prls, and \pgls\ in order. 
In general, we see that the stochastic block-coordinate algorithms converge faster than their deterministic counterparts across the board, even after thousands of iterations of the deterministic methods. Thus, block-coordinate methods seem to be better even at high precision, while simultaneously achieving better early performance due to the many more updates performed per unit of work.

\begin{figure}
        \centering
        \includegraphics[width=0.23\textwidth]{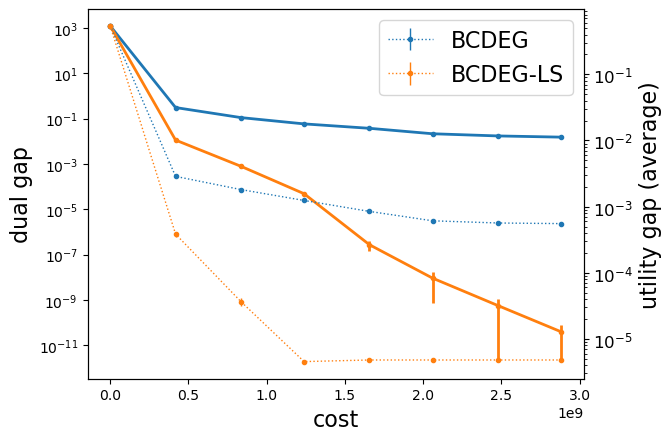}
        \includegraphics[width=0.23\textwidth]{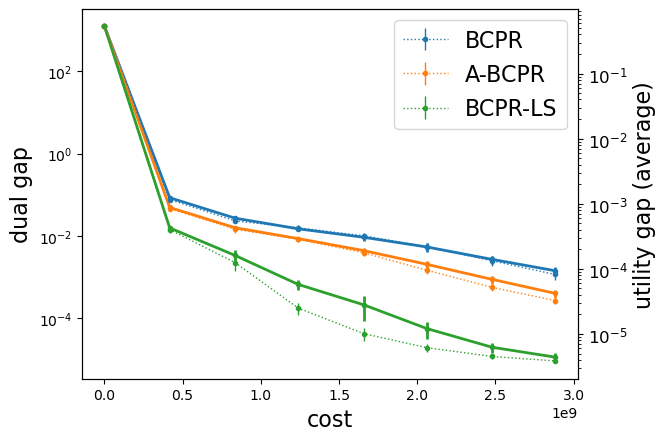}
        \includegraphics[width=0.23\textwidth]{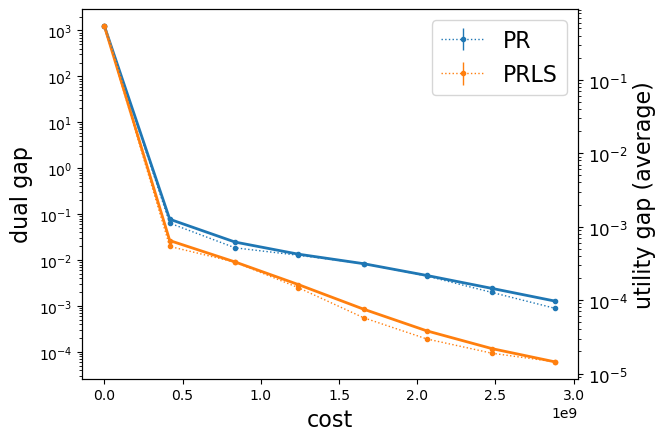}
        \includegraphics[width=0.23\textwidth]{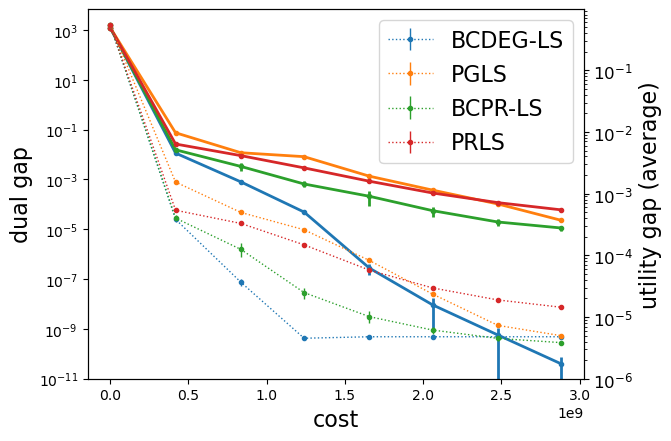}
    \caption{
        Performance on movierating instances. The plot setup is the same as in \cref{fig:compare-simulated-data}.
    }
    \label{fig:compare-real-data}
\end{figure}

\begin{figure}
        \hspace{2pt}
        \includegraphics[width=0.20\textwidth]{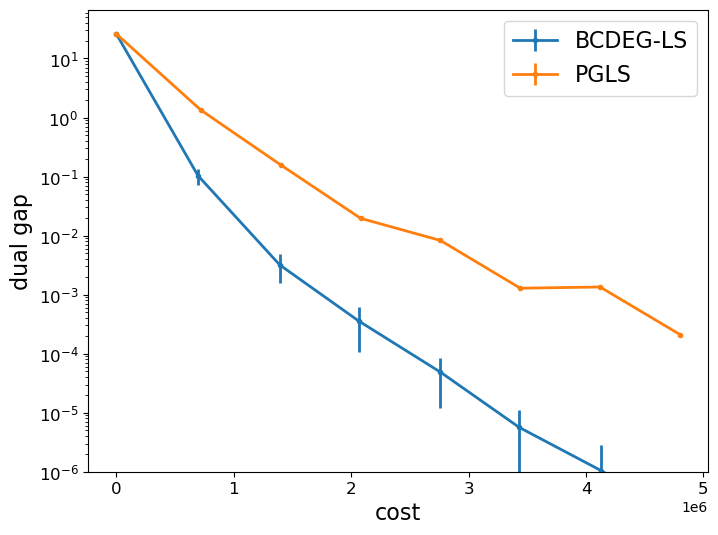} \hspace{11pt}
        \includegraphics[width=0.20\textwidth]{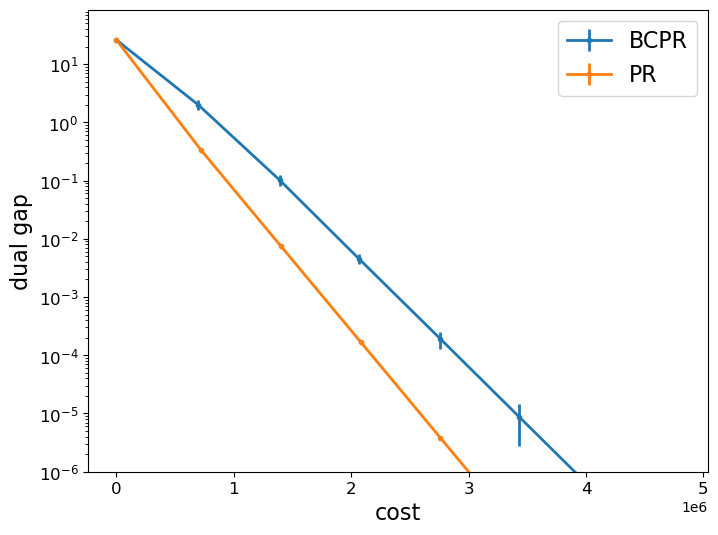}
    \caption{Performance on simulated instances with CES utility with $\rho \in (0, 1)$. The setup is the same as in \cref{fig:compare-simulated-data}.
    }
    \label{fig:compare-simulated-data-ces}
\end{figure}

\paragraph{CES utilities.}
Similar to linear utilities, we generate a $n = m = 200$ scale market instance with the CES-utility parameters generated as $v_{ij} = v_i v_j + \epsilon_{ij}$, where $v_i \sim \mathcal{N}(1, 0.2^2)$, $v_j \sim \mathcal{N}(1, 0.2^2)$, and $\epsilon_{ij} \sim \emph{uniform}(0, 0.2)$. 
For the CES instances we were not able to obtain high-accuracy solutions from existing conic solvers, and thus we only measure performance in terms of the dual gap.
Somewhat surprisingly, we find that for CES utilities vanilla \pr\ converges very fast; faster than \bcpr\ and all versions of \bcdeg. 
Moreover, \bcdeg\ suffers from extremely small stepsizes, so we have to use \bcdegls. 
Here, we used specific small parameters for the distributions in the simulated utilities. 
More discussion and experiments on CES utilities are given in \cref{app:ces}.

\section{Conclusion and Future Work}
We proposed two stochastic block-coordinate algorithms for computing large-scale ME: (proximal) block-coordinate descent on EG (\bcdeg) and block-coordinate proportional response (\bcpr). For each algorithm we provided  theoretical convergence guarantees and showed numerically that they outperform existing state-of-the-art algorithms.
We also provided a new economic interpretation of the projected gradient update used in {\bcdeg}.
For future work, we are interested in deriving a sublinear convergence rate for {\bcpr} with adaptive stepsizes, extending it to leverage distributed and parallel computing capabilities, and allowing more general dynamic settings and other buyer utility models.


\bibliography{refs.bib}


\newpage
\onecolumn
\appendix
\section{Block Coordinate Descent Algorithm for the EG Program}
\label{app:bcdeg}

\subsection{An efficient implementation of \bcdeg}

In order to implement \cref{algo:bcdeg} efficiently, we temporarily store a vector $u \in \RR^n$ in memory, which represents the current utilities for all buyers. 
Then, at each iteration $k$, we can dynamically update $u$ by only substituting $x_{\cdot j_k}$ while updating a particular block $j_k$, which can be done in $O(n)$ time. 
This avoids performing calculations of the form $u_i = \inp{v_i}{x_i}$, which require $O(nm)$ time when performed for each buyer.
See \cref{algo:efficient-bcdeg} for details. 
\begin{algorithm}
    \KwIn{Initial iterate $x^0$, stepsizes $\eta_1^k, \eta_2^k, \ldots, \eta_m^k,\; \forall\; k \in \NN$}
        $u_i^0 \gets \inp{v_i}{x_i^0}$ for all $i$\;
        \For{$k \gets 1, 2, \ldots $} {
            pick $j_k \in [m]$ with probability $1 / m$\;
            \For{$i \gets 1, 2, \cdots, n$}{
                \uIf{$u_i^{k-1} \ge \underline{u}_i$}{
                    $g_i^{k-1} \gets - B_i \log{u_i^{k-1}}$
                }
                \Else{
                    $g_i^{k-1} \gets - B_i \log{\underline{u}_i} - \frac{B_i}{\underline{u}_i} \left( u_i^{k-1} - \underline{u}_i \right) + \frac{B_i}{2\underline{u}_i^2} \left( u_i^{k-1} - \underline{u}_i \right)^2$
                }
            }
            $x^k_{\cdot j_k} \gets \text{Prox}(x^{k-1}_{\cdot j_k} - \eta_{j_k}^k g^{k-1})$\;
            $x^k_{\cdot j} \gets x^{k-1}_{\cdot j} \quad \forall j \neq j_k$\;
            $u^k \gets u^{k-1} + v_{j_k}({x^k_{\cdot j_k}} - {x^{k-1}_{\cdot j_k}})$\;
        }
    \caption{An efficient implementation of \bcdeg}
    \label{algo:efficient-bcdeg}
\end{algorithm}


\subsection{(Proximal) Block Coordinate Descent for the EG Program with Line Search}

We formally state \bcdegls\ here. This algorithm often outperforms others in our numerical experimental results. 
The difference to \bcdeg\ is that we check a descent condition (the if condition in \cref{algo:bcdeg-ls}) after performing a tentative proximal step. If the descent condition fails the check (the true case of the if-else statement), then we decrease the stepsize and perform another tentative step. If it passes the check then we commit to the tentative proximal step.
\begin{algorithm}
    \KwIn{Initial iterate $x^0$, initial stepsizes $\eta_1^0, \eta_2^0, \ldots, \eta_m^0$, $\rho_-, \rho_+ \in \RR^+$}
        \For{$k \gets 1, 2, \ldots $} {
            pick $j_k \in [m]$ with probability $1 / m$\;
            $g^{k-1} \gets \nabla_{\cdot j_k} \tilde{f}(x^{k - 1})$\;
            $(*)$ $x^+_{\cdot j_k} \gets \text{Prox}_{\Delta_{n}} \left( x^{k-1}_{\cdot j_k} - \eta_{j_k} g^{k-1} \right)$\;
            $g^+ \gets \nabla_{\cdot j_k} \tilde{f}(x^+)$\;
            \uIf{$\eta_{j_k} \lVert g^+ - g^{k-1} \rVert > \lVert x^+_{\cdot j_k} - x^{k-1}_{\cdot j_k} \rVert$}{
                $\eta_{j_k} \gets \max\left\{ \rho_- \eta_{j_k}, 1 / L_j \right\}$ where $L_j$ is computed by \eqref{eq:stepsize-bcd-eg}\;
                go back to $(*)$ and re-compute $x^+_{\cdot j_k}$ and $g^+$\;
            }
            \Else{
                $\eta_{j_k} \gets \rho_+ \eta_{j_k}$\;
                $x^k_{\cdot j_{k}} \gets x^+_{\cdot j_k}$ and $x^{(k)}_{\cdot j} \gets x^{(k-1)}_{\cdot j} \quad \forall\; j \neq j_{k}$\;                
            }
        }
    \caption{BCDEG-LS: (Proximal) Block Coordinate Descent for the EG Program with Line Search}
    \label{algo:bcdeg-ls}
\end{algorithm}

\subsection{Proof of \cref{lemma:coordinate-continuous}}
\label{app:stepsize-bcd-eg}

\begin{proof}

It can be verified by simple calculations that both $\tilde{f}_i(x)$ and $\nabla \tilde{f}_i(x)$ are continuous on their domains. 
Each entry of Hessian matrix of $\tilde{f}$ is of the form: 
\begin{equation*}
    \nabla_{ij}^2 \tilde{f}(x) = 
    \begin{cases} 
    {B_i v_{ij}^2}/{u_i^2(x)} & u_i(x) \ge \underline{u}_i \\ 
    {B_i v_{ij}^2}/{\underline{u}_i^2} & u_i(x) < \underline{u}_i. 
    \end{cases}
\end{equation*}

Let $x$ and $y$ be two vectors which only differ in one coordinate $j$. 
Then, 
\begin{equation}
    \left\lvert \nabla_{ij} \tilde{f}(y) - \nabla_{ij} \tilde{f}(x) \right\rvert = \left\lvert \nabla_{ij} \tilde{f}(y_i) - \nabla_{ij} \tilde{f}(x_i) \right\rvert \le \left\lvert \frac{B_i v_{ij}^2}{\underline{u}_i^2} \right\rvert \lvert y_{ij} - x_{ij} \rvert = L_j \lVert y_i - x_i \rVert, 
\end{equation}
where the inequality is due to 
\[
    \nabla_{ij}^2 \tilde{f}(x) \le \frac{B_i v_{ij}^2}{\underline{u}_i^2} \le \max_{i\in [n]}{\frac{B_i v_{ij}^{2}}{\underline{u}_i^2}} = L_j. 
\]

Therefore, 
\begin{equation} 
    \left\lVert \nabla_{\cdot j} \tilde{f}(y) - \nabla_{\cdot j} \tilde{f}(x) \right\rVert^2 = \sum_i \left\lvert \nabla_{ij} \tilde{f}(y) - \nabla_{ij} \tilde{f}(x) \right\rvert^2 \le L_j^2 \sum_i \lVert y_i - x_i \rVert^2 = L_j^2 \lVert y - x \rVert^2 = L_j^2 \lVert y_{\cdot j} - x_{\cdot j} \rVert^2, 
\label{eq:bcdeg-proof-smoothness-sqrt}
\end{equation}
that is, 
\begin{equation}
    \left\lVert \nabla_{\cdot j} \tilde{f}(y) - \nabla_{\cdot j} \tilde{f}(x) \right\rVert \le L_j \lVert y_{\cdot j} - x_{\cdot j} \rVert. 
\label{eq:bcdeg-proof-smoothness}
\end{equation}

To derive \eqref{eq:coordinate-continuous-3}, we can apply integration: 
\begin{align}
    \tilde{f}(y) - \tilde{f}(x) - \inp{\nabla_{\cdot j} \tilde{f}}{y - x} &= \int_0^1 \inp{\nabla_{\cdot j} \tilde{f}(x + \tau (y-x)) - \nabla_{\cdot j} \tilde{f}(x)}{y - x}\; d\tau \nonumber \\ 
    &\le \int_0^1 \lVert \nabla_{\cdot j} \tilde{f}(x + \tau (y-x)) - \nabla_{\cdot j} \tilde{f}(x) \rVert \lVert y - x \rVert\; d\tau \nonumber \\
    &\le \int_0^1 \tau L_j \lVert y - x \rVert^2 \; d\tau = \frac{L_j}{2} \lVert y - x \rVert^2. 
\end{align}

\end{proof}

\paragraph{Remark.} As a comparison, we provide a proof for a standard (global) Lipschitz constant $L$ as follows. 
Comparing this with the above proof, we show that when only one block of the variables $x_{\cdot j}$ change, the coordinate-wise Lipschitz constant $L_j$ is bounded by the maximum (over all $i$) diagonal value of the Hessian matrix. 
In contrast, for $L$ it depends on the maximum (over all $i$) sub-trace of the Hessian matrix. 

\begin{lemma}
For any $j \in [m]$, let 
\begin{align} 
    L = \max_{i\in [n]}{\frac{B_i \lVert v_i \rVert^2}{\underline{u}_i^2}}.
\label{eq:stepsize-pg} 
\end{align}
Then, for all $x, y \in \mathcal{X}$, we have
\begin{align}
    \tilde{f}(y) \le \tilde{f}(x) + \inp{\nabla \tilde{f}(x)}{y - x} + \frac{L}{2}{\lVert y - x \rVert^2}.  
    \label{eq:full-continuous-3}
\end{align}
\label{lemma:full-continuous}
\end{lemma}

\begin{proof}
Let $x$ and $y$ be two vectors in the domain of $\tilde{f}$, then 
    \begin{equation*}
        \left\lvert \nabla_{ij} \tilde{f}(y) - \nabla_{ij} \tilde{f}(x) \right\rvert
        = \left\lvert \nabla_{ij} \tilde{f}(y_i) - \nabla_{ij} \tilde{f}(x_i) \right\rvert
        \le \sum_{j'=0}^{n-1} \left\lvert \nabla_{ij} \tilde{f}(y^{j'}_i) - \nabla_{ij} \tilde{f}(x^{j'}_i) \right\rvert
        \le \sum_{j'=0}^{n-1} \left\lvert \frac{\partial^2}{\partial x_{ij} \partial x_{ij'}} \tilde{f}(x) \right\rvert \left\vert y_i^{j'} - x_i^{j'} \right\vert
    \end{equation*}
    where $y^{j'}_i$ is a vector whose first $(j' + 1)$'th components are copy of the first $(j' + 1)$'th components of $y_i$ and other components are copy of the last $(m - j' - 1)$'th components of $x_i$; 
    $x^{j'}_i$ is a vector whose first $j'$'th components are copy of the first $j'$'th components of $y_i$ and other components are copy of the last $(m - j')$'th components of $x_i$. 
    That is, each pair of $y^{j'}_i$ and $x^{j'}_i$ only differ in $(j' + 1)$'th component.  
    Also, we have $\vert y_i^{j'} - x_i^{j'} \vert \le \Vert y_i - x_i \Vert$ for all $i, j'$. 
    \begin{equation*}
        \sum_j \left\lvert \nabla_{ij} \tilde{f}(y) - \nabla_{ij} \tilde{f}(x) \right\rvert^2 
        \le \left( \sum_j \left\vert \nabla_{ij} \tilde{f}(y) - \nabla_{ij} \tilde{f}(x) \right\vert \right)^2
        \le \left( \sum_j \sum_{j'} \frac{B_i v_{ij} v_{ij'}}{\underline{u}_i^2} \left\Vert y_i - x_i \right\Vert \right)^2 \le L^2 \left\Vert y_i - x_i \right\Vert^2
    \end{equation*}
    since 
    \begin{equation*}
        \frac{\partial^2}{\partial x_{ij} \partial x_{ij'}} \tilde{f}(x) \le \frac{B_i v_{ij} v_{ij'}}{\underline{u}_i^2} \quad \text{and} \quad \sum_j \sum_{j'} \frac{B_i v_{ij} v_{ij'}}{\underline{u}_i^2} = \frac{B_i \lVert v_i \rVert^2}{\underline{u}_i^2} \le \max_{i \in [n]} \frac{B_i \lVert v_i \rVert^2}{\underline{u}_i^2} = L. 
    \end{equation*}
    Therefore, 
    \begin{equation} 
        \left\lVert \nabla \tilde{f}(y) - \nabla \tilde{f}(x) \right\rVert^2 = \sum_i \sum_j \left\lvert \nabla_{ij} \tilde{f}(y) - \nabla_{ij} \tilde{f}(x) \right\rvert^2 \le L^2 \sum_i \lVert y_i - x_i \rVert^2 = L^2 \lVert y - x \rVert^2. 
\end{equation}
\end{proof}
Note that this is a tighter bound than \citet{gao2020first} because $\max_{i \in [n]} \frac{B_i \lVert v_i \rVert^2}{\underline{u}_i^2} \le \max_{i \in [n]} \frac{B_i}{\underline{u}_i^2} \max_{i \in [n]} \lVert v_i \rVert^2$.

\subsection{Proof of Lemma~\ref{lemma:proximal-pl}}
\label{app:modified-proximal-pl-condition}

\begin{proof}

Since $g(u)$ is strongly convex, there exists a unique $u^*$ such that 
\[ 
    g(u^*) = g^* \quad\quad Ax^* = u^* \quad \forall x^* \in \mathcal{X}^*. 
\]

Thus, the set of optimal solutions $\mathcal{X}^*$ can be described by the following polyhedral set for some $C$: 
\[ 
    \mathcal{X}^* = \{ x^*: A x^* = u^*, C x^* \le d \}. 
\]

Assume that the optimal polyhedral set $\mathcal{X}^*$ is non-empty, then the Hoffman inequality~\citep{hoffman2003approximate} tells us there exists some positive constant 
such that
\begin{equation} 
    \lVert x - x^* \rVert \le \theta(A, C) 
    \begin{Vmatrix} 
        \begin{bmatrix}
        Ax - u^* \\
        \left( Cx - d \right)^+
        \end{bmatrix}	
    \end{Vmatrix} 
    \quad \forall\; x. 
\end{equation}
Let $\mathcal{X}$ be the set of feasible allocations. For any  $x \in \mathcal{X}$ such that $Cx \le d$, we have 
\begin{equation}
\lVert x - x^* \rVert \le \theta(A, C) \lVert Ax - Ax^* \rVert. 
\label{eq:x-bound-by-Ax}
\end{equation}

From strong convexity of $\tilde{g}$ and \eqref{eq:x-bound-by-Ax}, we have 
\begin{align}
F^* = F(x^*) & \ge F(x) + \inp{\nabla \tilde{f}(x)}{x^* - x} + \frac{\mu}{2} \lVert Ax^* - Ax \rVert^2 + \psi(x^*) - \psi(x) \nonumber \\ 
&\ge F(x) + \inp{\nabla \tilde{f}(x)}{x^* - x} + \frac{\mu}{2\theta^2(A, C)}\lVert x^* - x \rVert^2 + \psi(x^*) - \psi(x). 
\end{align}

If $L \ge \frac{\mu}{\theta^2(A, C)}$, then 
\begin{align}
F^* &\ge F(x) + \inp{\nabla \tilde{f}(x)}{x^* - x} + \frac{\mu}{2\theta^2(A, C)}\lVert x^* - x \rVert^2 + \psi(x^*) - \psi(x) \nonumber \\
&\ge F(x) + \min_{y \in \RR^{nm}} \left\{\inp{\nabla \tilde{f}(x)}{y - x} + \frac{\mu}{2\theta^2(A, C)} {\lVert y - x \rVert}^2 + \psi(y) - \psi(x) \right\} \nonumber \\
&= F(x) - \frac{\theta^2(A, C)}{2\mu}\mathcal{D}_{\psi}(x, \frac{\mu}{\theta^2(A, C)}). 
\end{align}
Hence, 
\begin{equation}
    \frac{1}{2}\mathcal{D}_{\psi}(x, \frac{\mu}{\theta^2(A, C)}) \ge \frac{\mu}{\theta^2(A, C)}(F(x) - F^*). 
\end{equation}
Since $\mathcal{D}(x, \cdot)$ is non-decreasing \citep[Lemma 1]{karimi2016linear}, we have
\begin{equation}
    \frac{1}{2}\mathcal{D}_{\psi}(x, L) \ge \frac{\mu}{\theta^2(A, C)}(F(x) - F^*). 
\end{equation}

If $L < \frac{\mu}{\theta^2(A, C)}$, then 
\begin{align}
F^* &\ge F(x) + \inp{\nabla \tilde{f}(x)}{x^* - x} + \frac{\mu}{2\theta^2(A, C)}\lVert x^* - x \rVert^2 + \psi(x^*) - \psi(x) \nonumber \\
&\ge F(x) + \inp{\nabla \tilde{f}(x)}{x^* - x} + \frac{L}{2}\lVert x^* - x \rVert^2 + \psi(x^*) - \psi(x) \nonumber \\
&\ge F(x) + \min_{y \in \RR^{nm}} \left\{\inp{\nabla \tilde{f}(x)}{y - x} + \frac{L}{2} {\lVert y - x \rVert}^2 + \psi(y) - \psi(x) \right\} \nonumber \\
&= F(x) - \frac{1}{2L}\mathcal{D}_{\psi}(x, L). 
\end{align}
Hence, 
\begin{equation} 
    \frac{1}{2}\mathcal{D}_{\psi}(x, L) \ge L(F(x) - F^*).  
\end{equation}

Combining the two cases, we can conclude \eqref{eq:proximal-pl-condition}. 
\end{proof}

\subsection{Proof of Theorem~\ref{theorem:bcdeg-convergence}}
\label{app:convergence-bcd-eg}

The following convergence analysis is typical for uniform (block) coordinate descent algorithms. Leveraging \eqref{eq:proximal-pl-condition}, it is easy to derive a linear convergence rate for both \bcdeg\ and \bcdegls. This proof scheme was given in \citet[Appendix H]{karimi2016linear}, but here we give a variant with different stepsizes across blocks and iterations. 

\begin{proof}
    To remind readers of notations, we define (again)
    \begin{equation}
        T_j(x) = \arg\min_{y_{\cdot j}} \inp{\nabla_{\cdot j} f(x)}{y_{\cdot j} - x_{\cdot j}} + L_j \lVert y_{\cdot j} - x_{\cdot j} \rVert^2 + \psi_j(y_{\cdot j})
    \end{equation}
    for \bcdeg\ and 
    \begin{equation}
        T_j^k(x) = \arg\min_{y_{\cdot j}} \inp{\nabla_{\cdot j} f(x)}{y_{\cdot j} - x_{\cdot j}} + \frac{1}{\eta_j^k} \lVert y_{\cdot j} - x_{\cdot j} \rVert^2 + \psi_j(y_{\cdot j})
    \end{equation}
    for \bcdegls. 
    
    Then, we have (e.g. for \bcdeg)
    \begin{align}
    \expect { F(x^{k+1}) | x^k } &\le F(x^k) + \frac{1}{m} \sum_j \left\{\inp{\nabla_{\cdot j} \tilde{f}(x^k)}{T_j(x) - x_{\cdot j}^k } + \frac{L_j}{2} \lVert T_j(x) - x_{\cdot j}^k \rVert^2 + \psi_j(T_j(x)) \right\} \label{eq:bcdeg-proof-smoothness-in-thm} \\
        &= \tilde{f}(x^k) + \frac{1}{m} \sum_j \min_{y_{\cdot j}} \left\{\inp{\nabla_{\cdot j} \tilde{f}(x^k)}{y_{\cdot j} - x_{\cdot j}^k} + \frac{L_j}{2} \lVert y_{\cdot j} - x_{\cdot j}^k \rVert^2 + \psi_j(y_{\cdot j}) \right\} \nonumber \\
        &= \tilde{f}(x^k) + \frac{1}{m} \min_y \left\{\inp{\nabla \tilde{f}(x^k)}{y - x^k} + \frac{1}{2} {\lVert y - x^k \rVert}_L^2 + \psi(y) \right\} \nonumber \\
        &= F(x^k) + \frac{1}{m} \min_y \left\{\inp{\nabla \tilde{f}(x^k)}{y - x^k} + \frac{1}{2} {\lVert y - x^k \rVert}_L^2 + \psi(y) - \psi(x^k) \right\} \nonumber \\
        &\le F(x^k) + \frac{1}{m} \min_y \left\{\inp{\nabla \tilde{f}(x^k)}{y - x^k} + \frac{L_{\max}}{2} \lVert y - x^k \rVert^2 + \psi(y) - \psi(x^k) \right\} \nonumber \\ 
        &= F(x^k) -\frac{1}{2 m L_{\max}} \mathcal{D}_\psi(x^k, L_{\max}), 
        \label{eq:bcdeg-proof-pl-condition-in-thm}
    \end{align}
    where we used Lipschitz smoothness (or line search condition for \bcdegls) of $\tilde{f}$, the definition of $T_j(x)$, and separability on each block $j$. 

    
    

    For \bcdegls, replacing \eqref{eq:bcdeg-proof-smoothness-in-thm} with the following inequality: 
    \begin{equation}
        \expect { F(x^{k+1}) | x^k } \le F(x^k) + \frac{1}{m} \sum_j \left\{\inp{\nabla_{\cdot j} \tilde{f}(x^k)}{T_j(x) - x_{\cdot j}^k } + \frac{1}{2 \eta_j^k} \lVert T_j(x) - x_{\cdot j}^k \rVert^2 + \psi_j(T_j(x)) \right\}, 
    \end{equation}
    and by $1 / \eta_j^k \le L_j \le L_{\max}$ (used in the last inequality) we have \eqref{eq:bcdeg-proof-pl-condition-in-thm}. 

    By \cref{lemma:proximal-pl}, we have 
    \begin{equation}
        \expect { F(x^{k+1}) | x^k } - F(x^k) \le -\frac{1}{2 m L_{\max}} \mathcal{D}_\psi(x^k, L_{\max}) \le - \min\left\{\frac{\mu}{m L_{\max} \theta^2(A, C)}, \frac{1}{m} \right\} (F(x^k) - F^*). 
    \end{equation} 
    
    Subtracting $F^*$ on the both sides, Rearranging the above inequality, and taking expectation w.r.t. $j_1, \ldots, j_k$, we obtain \eqref{eq:linear-convergence-bcd-eg} by induction. 
\end{proof}

\section{Algorithm Details for Projection onto a Simplex}
\label{app:econ-projection}

We show a key theorem for projection onto a simplex algorithm and a formally stated algorithm as background for readers who are not familiar with this algorithm. 

\begin{theorem}{\citep[Theorem 2.2.]{chen2011projection}}
    For any vector $y \in \mathbb{R}^n$, the projection of $y$ onto $\Delta^n$ is obtained by the positive part of $y - \hat{t}$: 
    \begin{equation} 
        x = (y - \hat{t})^+,  
        \label{eq:x} 
    \end{equation}
    where $\hat{t}$ is the only one in $\{ t_i: i=0, \ldots , n-1 \}$ that falls into the corresponding interval as follows, 
    \begin{equation} 
        t_i := \frac{\sum_{j=i+1}^{n}{y_{(j)} - 1}}{n - i}, i=0, \ldots, n-1, \quad \text{ where } t_1 \le y_{(1)} \text{ and } y_{(i)} \le t_{i} \le y_{(i+1)}, i=1, \ldots, n-1. 
        \label{eq:t_i} 
    \end{equation} 
\end{theorem}

Hence, to find $x$, we only need to find the $t_i$ in~\eqref{eq:t_i} that falls into the corresponding interval, 
claim it as the optimal $\hat{t}$, 
and then compute $x$ based on~\eqref{eq:x}. 
The following algorithm is implied by the above theorem. 

\begin{algorithm}[H]
    \KwIn{$y = (y_1, y_2, \ldots, y_n)^\top \in \mathbb{R}^n$}
    \KwOut{$x = (x_1, x_2, \ldots, x_n)^\top \in \mathbb{R}^n$ as the projection of $y$ onto $\Delta^n$}
    Sort $y$ in the ascending order as $y_{(1)} \le \cdots \le y_{(n)}$\;
    $\hat{t} \gets \frac{\sum_{j=1}^{n}{y_{(j)}} - 1}{n}$\;   
    \For{$i \gets n-1, \ldots, 1$}{
        $t_{i} \gets \frac{\sum_{j=i+1}^{n}{y_{(j)}} - 1}{n-i}$\;
        \uIf{$t_i \ge y_{(i)}$} {
            $\hat{t} \gets t_i$ and break the loop\;
        }
    }
    $x \gets (y - \hat{t})^+$\;
    \caption{Projection onto a Simplex}
    \label{algo:projsplx}
\end{algorithm}

\subsection{Discussion on alternative algorithms for projection onto a simplex}

In this paper, we assume that the above algorithm is an efficient method to project an array onto a simplex, and interpret projection-type ME-solving algorithms based on this algorithm. 
There are several variants of this type of projection algorithm, and some of them give improved complexity results or practical performance.
\citet{condat2016fast} summarize most of these algorithms, which use different methods to find the correct threshold value. 
Some of these algorithms exploit problem structure to obtain $O(n)$ expected or practical complexity. However, in the worst case, there is still no better result than $O(n\log(n))$. 
\citet{perez2020filtered} proposed a method with a worst-case linear-time complexity result. 
Their result is analogous to bucket sort, it assumes that the set of possible values that you might encounter is constant. 

\section{Block Coordinate Proportional Response}
\label{app:rbcd}

\subsection{Block Coordinate Proportional Response with Line Search}

We formally state Block Coordinate Proportional Response with Line Search (\bcprls) algorithm as follows. 
Note that $\delta \in [0, 1)$ is a ``conservative'' factor - we can set $\delta > 0$ to use more conservative stepsize strategy than standard line search ($\delta = 0$). 
In $\delta > 0$ case, we can guarantee some (weak) convergence property for \bcprls. 

\begin{algorithm}[H]
    \KwIn{$b^0 \in \RR^{n\times m}, p^0 \in \RR^m, \alpha_1, \ldots, \alpha_n \in \RR^+, \rho_-, \rho_+ \in \RR^+$}
        \For{$k \gets 1, 2, \ldots $} {
            pick $i_k \in [n]$ with probability $1/n$\; 
            (*) compute $b_{i_k}^+$ based on \eqref{eq:bcpr-update-simple} with stepsize $(1 - \delta)\alpha_{i_k}$\;
            $p^+ \gets p + b_{i_k}^+ - b_{i_k}$\;
            \uIf{$\alpha_{i_k} \DKL(p^+, p) > \DKL(b_{i_k}^+, b_{i_k})$}{
                $\alpha_{i_k} \gets \rho_- \alpha_{i_k}$\;
                go back to (*) and re-compute $b_{ik}^+$ and $p^+$\;
            }
            \Else{
                $\alpha_{i_k} \gets \rho_+ \alpha_{i_k}$\;
                $b_{i_k}^k \gets b_{i_k}^+$ and $b_{i'}^k \gets b_{i'}^{k-1} \quad \forall\; i' \neq i_{k}$\;
            }
        }
\caption{Block Coordinate Proportional Response (BCPR) with Line Search}
\label{algo:bcpr-ls}
\end{algorithm}

\subsection{Proof of \cref{eq:bcpr-update-simple}}
\label{app:update-rule-rbcd}

\begin{proof}
    The original minimization problem can be solved by Lagrangian method by introducing a Lagrangian multiplier $\lambda$: 
    \begin{equation}
        \min\mathcal{L}(a_{1}, a_{2}, \ldots, a_{m}, \lambda) = {\inp{\nabla_{i} \varphi(b)}{a - b_{i}} + \frac{1}{\alpha} D_{\text{KL}}(a, b_{i}) + \lambda \Big(B_{i} - \sum_{j}{a_{j}}\Big)}. 
        \label{rcd:lagrangian}
    \end{equation}
    Let $\nabla_{a_{1}, a_{2}, \ldots, a_{m}, \lambda}\mathcal{L} = 0$, we have 
    \begin{equation}
        \frac{\partial}{\partial a_{j}} \bigg( \inp{\nabla_{i} \varphi(b)}{a - b_{i}} + \frac{1}{\alpha} \DKL(a, b_i) \bigg) = \lambda \quad \forall j =1, 2, \ldots, m \quad \mbox{and} \quad \sum_{j}{a_{j}} = B_i. 
    \end{equation}
    The partial derivatives are equal to 
    \[ 1 - \log{\left( \frac{v_{ij}}{p_{j}} \right)} + \frac{1}{\alpha} \left( 1 + \log{\frac{a_j}{b_{ij}}} \right) \quad \forall j =1, 2, \ldots, m. \]
    Let $a_{j} = \frac{1}{Z} b_{ij} \Big(\frac{v_{ij}}{p_{j}}\Big)^{\alpha}$, $\lambda = 1 + \frac{1}{\alpha} - \log{Z}$, where $Z = \frac{1}{B_{i}}\sum_{j}{b_{ij} \Big(\frac{v_{ij}}{p_{j}}\Big)^{\alpha}}$. 
    We can verify that 
    \begin{equation}
        1 + \frac{1}{\alpha} - \log{Z} = \frac{\partial}{\partial a_j} \left( \inp{\nabla_i \varphi(b)}{a - b_i} + \frac{1}{\alpha} \DKL(a, b_{i}) \right) = \lambda \quad \mbox{and} \quad \sum_j{a_j} = \frac{1}{Z}\sum_j{b_{ij} \left(\frac{v_{ij}}{p_j}\right)^\alpha} = B_i. 
    \end{equation} 
    Therefore, $(a_1, a_2, \ldots, a_m, \lambda)$ is a stationary point of the Lagrange function. 
    $a = (a_1, a_2, \ldots, a_m)$ is a solution to the original minimization problem due to the convexity of the problem. 
\end{proof}

\subsection{Coordinate-wise relative smoothness (in expectation)}

In the rest of this section, except for \abcpr, we use general Bregman distance $D(a, b)$ between $a, b \in \RR^n$ and $D_L(a, b) = \sum_i L_i D_i(a_i, b_i)$ where $D_i(a_i, b_i)$ is coordinate-wise Bregman distance and $L \in \RR^n$. 
We also introduce $T^k = \sum_i T_i^k$ as below (similar to $l_2$ norm block coordinate descent): 
\begin{equation}
    T^k = \sum_{i}{T_i^k} = \argmin_{a \in \RR^{n \times m}} \left\{ \inp{\nabla \varphi(b)}{a - b} + \sum_i \frac{1}{\alpha_i^k} D_i(a_i, b_i) + r(a) \right\}. 
    \label{eq:rcd-update-full}
\end{equation}

Then, it is standard to have the following lemma. 

\begin{lemma}
    If $b^k$ is the random iterates generated by Algorithm~\ref{algo:bcpr-ls}, then 
    \begin{align}
        \expect* { \varphi(b^{k+1}) | b^k } \le& \varphi(b^k) + \frac{1}{n}\inp{\nabla \varphi(b^k)}{T^k - b^k} + \frac{1}{n} D_{1/\alpha^k}(T^k, b^k) \label{eq:h-ESO} \\
        =& \varphi(b^k) + \expect*{\inp{\nabla \varphi(b^k)}{b^{k+1} - b^k} | b^k } + \expect*{D_{1/\alpha^k}(b^{k+1}, b^k) | b^k }. 
        \label{eq:rcd-relative-smooth}
    \end{align}
    \label{lemma:relative-smooth}
\end{lemma}
Note that \eqref{eq:h-ESO} is a special case (the cardinality of sampling equals to $1$, uniform sampling) of $h$-ESO assumption defined in \citet{hanzely2021fastest}. Essentially, this is coordinate-wise relative smoothness in expectation. 
\begin{proof}
    Given any coordinate $i$, we have (let $b^+$ and $p^+$ be the next iterates of $b$ and $p$, respectively)
    \begin{align}
        \varphi(b^+) - \varphi(b) - \inp{\nabla_i{\varphi(b)}}{b_i^+ - b_i} =& \varphi(b^+) - \varphi(b) - \inp{\nabla{\varphi(b)}}{b^+ - b} \nonumber \\
        =& - \sum_{i, j}{b_{ij}^{+}\log{\frac{v_{ij}}{p_{j}^{+}}}} + \sum_{i, j}{b_{ij}\log{\frac{v_{ij}}{p_{j}}}} - \sum_{i, j}{\big( 1 - \log{\frac{v_{ij}}{p_{j}}} \big)(b_{ij}^{+} - b_{ij})} \nonumber \\
        =& \sum_{i, j}{b_{ij}^{+} \log{\frac{p_{j}^{+}}{p_{j}}}} = \sum_{j}{p_{j}^{+} \log{\frac{p_{j}^{+}}{p_{j}}}}, 
    \end{align}
    and $\DKL(b_i^+, b_i) = \sum_{j}{b_{ij}^{+} \log{\frac{b_{ij}^{+}}{b_{ij}}}}$. 
    The line search condition ensures 
    \begin{align}
        \varphi(b^+) \le \varphi(b) + \inp{\nabla_i{\varphi(b)}}{b_i^+ - b_i} + \frac{1}{\alpha_i} D_{\text{KL}}(b_i^+, b_i).  \label{eq:relative-smooth-i}
    \end{align}

    Note that if we let $a \in \mathbb{R}$ and $f(x) = x\log{\frac{x}{x - a}}$, then $f'(x) = \log{(1 + \frac{a}{x - a})} - \frac{a}{x - a} \le \frac{a}{x - a} - \frac{a}{x - a} = 0$. Since $p_{j}^{+} - p_{j} = b_{ij}^{+} - b_{ij}$, let $a = b_{ij}^{+} - b_{ij}$, we have 
    \begin{equation}
        \sum_j{p_{j}^{+} \log{\frac{p_{j}^{+}}{p_{j}}}} = \sum_j f(p_{j}^{+}) \le \sum_j f(b_{ij}^{+}) = \sum_j {b_{ij}^{+} \log{\frac{b_{ij}^{+}}{b_{ij}}}}, 
    \end{equation}
    where we use $p_j^+ \ge b_{ij}^+$ and monotonity of $f(x)$. 
    Note that \eqref{eq:relative-smooth-i} always holds when $\alpha_i = 1$ for all $i$. 

    Hence, by separability on each block of coordinates we have 
    \begin{align}
        \expect*{ \varphi(b^{k+1}) | b^k } \le& \sum_i \frac{1}{n} \left\{ \varphi(b^k) + \inp{\nabla_i{\varphi(b)}}{T_i^k - b_i^k} + \frac{1}{\alpha_i^k} \DKL(T_i^k, b_i^k) \right\} \nonumber \\ 
        =& \varphi(b^k) + \frac{1}{n}\inp{\nabla \varphi(b^k)}{T^k - b^k} + \frac{1}{n} D_{1/\alpha^k}(T^k, b^k). 
    \end{align}

    As we know
    \begin{equation}
        \expect*{\inp{\nabla \varphi(b^k)}{b^{k+1} - b^k} | b^k } = \sum_i \frac{1}{n} \inp{\nabla_i \varphi(b^k)}{T_i^k - b_i^k} = \frac{1}{n}\inp{\nabla \varphi(b^k)}{T^k - b^k}
        \label{eq:exp-t-relation-1}
    \end{equation}
    and 
    \begin{equation}
        \expect*{ D_{1/\alpha^k}(b^{k+1}, b^k) | b^k } = \sum_i \frac{1}{n} \frac{1}{\alpha_i^k} D_i(T_i^k, b_i^k) = \frac{1}{n} D_{1/\alpha^k}(T^k, b^k), 
        \label{eq:exp-t-relation-2}
    \end{equation}
    \eqref{eq:h-ESO} and \eqref{eq:rcd-relative-smooth} are equivalent. 

\end{proof}

\subsection{A useful descent lemma}

Next, we will consider \bcpr\ with line search (\bcprls). We can directly see that \bcpr\ (with fixed stepsizes)'s sublinear convergence rate is a special case of our convergence property. 

The following lemma can be derived from Lemma~\ref{lemma:relative-smooth} and three point property (e.g. \citet[Lemma 3.5]{hanzely2021fastest}). 
Beyond existing results, we prove a version for line search and $\alpha_i^k = (1 - \delta)/L_i^k$ stepsize strategy. 

\begin{lemma}
If $b^k$ is the random iterates generated by Algorithm~\ref{algo:bcpr-ls}, then 
\begin{align}
    \expect*{ \Phi(b^{k+1}) | b^k } &\le \frac{n-1}{n} \Phi(b^k) + \frac{1}{n} \Big\{ \Phi(u) + D_{\alpha^k}(u, b^k) - D_{\alpha^k}(u, T^k) - \delta D_{\alpha^k}(T^k, b^k) \Big\} \nonumber \\ 
    &= \frac{n-1}{n} \Phi(b^k) + \frac{1}{n} \Phi(u) + \frac{1}{n} D_{\frac{1}{\alpha^k}}(u, b^k) - \frac{1}{n} \expect*{ D_{\frac{1}{\alpha^k}}(u, b^{k+1}) | b^k } - \frac{\delta}{n} \expect*{ D_{\frac{1}{\alpha^k}}(b^{k+1}, b^k) | b^k }, 
    \label{eq:bcpr-descent}
\end{align}
where $T^k$ defined in \eqref{eq:rcd-update-full} and $u \in \dom \varphi \cup \dom h$. 
\label{lemma:bcpr-descent}
\end{lemma}

\begin{proof}


    

    Given any $b^k$ (iterate at the $k$th iteration), we have
    \begin{align}
        \expect*{\Phi(b^{k+1})|b^k} \le& \varphi(b^k) + \expect*{ \inp{\nabla \varphi(b^k)}{b^{k+1} - b^k} | b^k } + \expect*{ D_{L^k}(b^{k+1}, b^k) | b^k } + \expect*{ r(b^{k+1}) | b^k } \nonumber \\
        =& \varphi(b^k) + \frac{1}{n}\inp{\nabla \varphi(b^k)}{T^k - b^k} + \frac{1}{n}D_{L^k}(T^k, b^k) + \frac{1}{n} \sum_i r_i\left( T_i^k \right) + \frac{1}{n} \sum_i \sum_{i' \neq i} r_{i'}\left( b_{i'}^k \right) \nonumber \\
        =& \varphi(b^k) + \frac{1}{n}\inp{\nabla \varphi(b^k)}{T^k - b^k} + \frac{1}{n}D_{L^k}(T^k, b^k) + \frac{1}{n} r\left( T^k \right) + \frac{n-1}{n} r\left( b^k \right) \nonumber \\
        =& \frac{n-1}{n} \Phi(b^k) + \frac{1}{n} \left\{ \varphi(b^k) + \inp{\nabla \varphi(b^k)}{T^k - b^k} + D_{L^k}(T^k, b^k) + r\left( T^k \right) \right\}, 
        \label{step-1}
    \end{align} 
    which follows from Lemma~\ref{lemma:relative-smooth} and separability of $r$. 

    
    By the optimality of \eqref{eq:rcd-update-full}, we have 
    \begin{equation}
        v^k = - \nabla \varphi(b^k) - \sum_i \frac{1}{\alpha_i^k} \nabla h_i(T_i^k) + \sum_i \frac{1}{\alpha_i^k} \nabla h_i(b^k_i)
        \label{eq:optimal-con}
    \end{equation}
    for some $v^k \in \partial r(T^k)$. 
    
    By convexity of $\varphi$ and $r$, for any $u \in \dom \varphi \cup \dom h$, 
    \begin{equation}
        r(T^k) \le r(u) - \inp{v^k}{u - T^k}, 
        \label{eq:cvx-r}
    \end{equation}
    \begin{equation}
        \varphi(b^k) \le \varphi(u) - \inp{\nabla \varphi(b^k)}{u - b^k}. 
        \label{eq:cvx-varphi}
    \end{equation}

    By the definition of Bregman distance we have so-called three point identity:  
    \begin{equation}
        \inp{\nabla h_i(T_i^k) - \nabla h_i(b^k_i)}{u_i - T_i^k} = D_i(u_i, b^k_i) - D_i(u_i, T_i^k) - D_i(T_i^k, b^k_i)
        \label{eq:three-point}
    \end{equation}
    for all $u_i \in \dom h_i$ and all $i \in [n]$. 

    Then, by combining all above inequalities and equalities we have 
    \begin{align}
        \expect*{ \Phi(b^{k+1}) | b^k } \le& \frac{n-1}{n} \Phi(b^k) + \frac{1}{n} \left\{ \varphi(b^k) + \inp{\nabla \varphi(b^k)}{T^k - b^k} + D_{L^k}(T^k, b^k) + r\left( T^k \right) \right\} \nonumber \\
        \le& \frac{n-1}{n} \Phi(b^k) + \frac{1}{n} \left\{ \varphi(b^k) + \inp{\nabla \varphi(b^k)}{u - b^k} \right\} \nonumber \\ 
        &\quad + \frac{1}{n} \left\{D_{L^k}(T^k, b^k) + r\left( u \right) + \sum_i \frac{1}{\alpha_i^k} \inp{\left( \nabla h_i \left( T_i^k \right) - \nabla h_i \left( b_i^k \right) \right)}{u - T^k} \right\} \nonumber \\
        \le& \frac{n-1}{n} \Phi(b^k) + \frac{1}{n} \varphi(u) + \frac{1}{n} \left\{D_{L^k}(T^k, b^k) + r\left( u \right) + \sum_i \frac{1}{\alpha_i^k} \inp{\left( \nabla h_i \left( T_i^k \right) - \nabla h_i \left( b_i^k \right) \right)}{u - T^k} \right\} \nonumber \\
        =& \frac{n-1}{n} \Phi(b^k) + \frac{1}{n}\left\{ \Phi(u) + D_{\alpha^k}(u, b^k) - D_{\alpha^k}(u, T^k) - \delta D_{\alpha^k}(T^k, b^k) \right\}
        \label{eq:bcpr-ls-iteration-decrease}
    \end{align}
    for any $u$. Here, we use \eqref{step-1}, \eqref{eq:cvx-r} with \eqref{eq:optimal-con}, \eqref{eq:cvx-varphi}, \eqref{eq:three-point} and $\alpha_i^k = (1 - \delta) / L_i^k$ for $\delta \in [0, 1)$. 
    \eqref{eq:bcpr-descent} follows from identity relations similar to \eqref{eq:exp-t-relation-1} and \eqref{eq:exp-t-relation-2}. 
    \end{proof}
    
    The following corollaries provide some (weak) convergence properties for \bcprls. 
    
    \begin{corollary}
        If $b^k$ is the random iterates generated by Algorithm~\ref{algo:bcpr-ls}, then
        \begin{enumerate}
            \item[(i)] $\EE[\Phi(b^k)]$ is non-increasing; 
            \item[(ii)] $\sum_{k=0}^\infty \EE[D_{1/\alpha^k}(b^{k+1}, b^k)] < \infty$, $\EE[D_i(b_i^{k+1}, b_i^k)] \rightarrow 0$ as $k \rightarrow \infty$ for all $i$;  
            \item[(iii)] $\sum_{k=0}^\infty \EE[D_{1/\alpha^k}(b^k, b^{k+1})] < \infty$, $\EE[D_i(b_i^k, b_i^{k+1})] \rightarrow 0$ as $k \rightarrow \infty$ for all $i$;
            \item[(iv)] $\EE[D_{1/\alpha^k}(b^*, b^k)] \le \EE[D_{1/\alpha^k}(b^*, b^{k+1})]$ as $k \rightarrow \infty$. 
        \end{enumerate}
    \end{corollary}

    \begin{proof}
    \item[(i)] Let $u = b^k$ in \eqref{eq:bcpr-descent}, we have 
    \begin{equation}
        \expect*{ \Phi(b^{k+1}) | b^k} \le \Phi(b^k) 
        \label{eq:bcpr-ls-objective-descent-in-exp}
    \end{equation}
    by using non-negativity of Bregman distance. This shows objective descent in expectation. 
    
    \item[(ii)] Also, we have 
    \begin{equation}
        \EE[ \Phi(b^{k+1}) \vert b^k ] \le \Phi(b^k) + \frac{1}{n}\left\{ - D_{\alpha^k}(b^k, T^k) - \delta D_{\alpha^k}(T^k, b^k) \right\} \le \Phi(b^k) - \frac{\delta}{n} D_{\alpha^k}(T^k, b^k). 
        \label{eq:star}
    \end{equation}

    As we know $\EE[D_{\alpha^k}(b^{k+1}, b^k) \vert b^k] = \frac{1}{n} D_{\alpha^k}(T^k, b^k)$, 
    \eqref{eq:star} is equivalent to 
    \begin{equation}
        \delta \EE[D_{\alpha^k}(b^{k+1}, b^k) \vert b^k] \le \Phi(b^k) - \EE[ \Phi(b^{k+1}) \vert b^k ]. 
        \label{eq:obj-bound-d}
    \end{equation}

    Summing up \eqref{eq:obj-bound-d} over $0, 1, \ldots, k-1$ and taking expectation w.r.t. $i_0, i_1, \ldots, i_{k-1}$, we have
    \begin{equation}
        \delta \sum_{l=0}^k \EE[D_{\alpha^l}(b^{l+1}, b^l)] \le \Phi(b^0) - \EE[\Phi(b^k)] < \infty. 
        \label{eq:sum-bound}
    \end{equation}

    Let $k \rightarrow \infty$, we have $\sum_{l=0}^\infty \EE[D_{\alpha^l}(b^{l+1}, b^l)] < \infty$, which implies $\EE[D_{\alpha^k}(b^{k+1}, b^k)] \rightarrow 0$ as $k \rightarrow \infty$. 

    \item[(iii)] Similar to \eqref{eq:star}, we have 
    \begin{equation}
        \EE[ \Phi(b^{k+1}) \vert b^k ] \le \Phi(b^k) + \frac{1}{n}\left\{ - D_{\alpha^k}(b^k, T^k) - \delta D_{\alpha^k}(T^k, b^k) \right\} \le \Phi(b^k) - \frac{1}{n} D_{\alpha^k}(b^k, T^k). 
        \label{eq:star-}
    \end{equation}
    Then, (iii) follows from the same argument as the proof of (ii).

    \item[(iv)] Because $\EE[\Phi(b^k)]$ is non-increasing and $\EE[\Phi(b^k)] \ge \Phi^* > -\infty$, $\EE[\Phi(b^{k+1}) - \Phi(b^k)] \rightarrow 0$. 

    Let $u = b^*$ in \eqref{eq:bcpr-ls-iteration-decrease}, we have
    \begin{equation}
        \expect*{ \Phi(b^{k+1}) | b^k } \le \frac{n-1}{n} \Phi(b^k) + \frac{1}{n}\left\{ \Phi(b^*) + D_{\alpha^k}(b^*, b^k) - D_{\alpha^k}(b^*, T^k) \right\}
        \label{eq:bcpr-ls-iteration-decrease-2}
    \end{equation}
    due to the non-negativity of Bregman distance. 
    Rearranging \eqref{eq:bcpr-ls-iteration-decrease-2}, summing it up over $0, \ldots, k$ and using the tower property, we have 
    \begin{equation}
        (n-1) \EE[\Phi(b^{k+1}) - \Phi(b^k)] + \left( \EE[\Phi(b^{k+1})] - \Phi(u) \right) \le \EE[D_{\alpha^k}(u, b^k)] - \EE[D_{\alpha^k}(u, b^{k+1})]. 
        \label{eq:bcpr-ls-iteration-decrease-3}
    \end{equation}
    When $k \rightarrow \infty$, the left hand side of \eqref{eq:bcpr-ls-iteration-decrease-3} goes to be non-negative, which imples (iv). 
    \end{proof}

    


    \subsection{Proof of \cref{thm:bcpr-fixed-stepsize}}
    
    \begin{proof}

    For any $u \in \dom h$, we have 
    \begin{align}
        \expect*{ D_{1/\alpha^k}(u, b^{k+1}) | b^k }
        &= \sum_i \frac{1}{n} \left\{ \frac{1}{\alpha_i^k}D_i(u_i, T_i^k) + \sum_{i' \neq i} \frac{1}{\alpha_{i'}^k}D_i(u_{i'}, b_{i'}^k) \right\} \nonumber \\
        &= \frac{1}{n} D_{1 /\alpha^k}(u, T^k) + \frac{n - 1}{n} D_{1/\alpha^k}(u, b^k). 
        \label{step-3}
    \end{align}

    By \cref{lemma:bcpr-descent} and \eqref{step-3}, let $u = b^*$, then we get 
    \begin{align}
        \expect*{ \Phi(b^{k+1}) + D_{1/\alpha^k}(b^*, b^{k+1}) | b^k } \le \frac{n - 1}{n}\Phi(b^k) + \frac{1}{n}\Phi(b^*) + D_{1/\alpha^k}(b^*, b^k). 
        \label{step-4-}
    \end{align}
    Rearranging \eqref{step-4-}, we attain  
    \begin{align}
        \expect*{ \Phi(b^{k+1}) - \Phi(b^*) + D_{1/\alpha^k}(b^*, b^{k+1}) | b^k } \le \frac{n - 1}{n}\left( \Phi(b^k) - \Phi(b^*) \right) + D_{1/\alpha^k}(b^*, b^k). 
        \label{step-4}
    \end{align}

    We take the expectation of~\eqref{step-4} with respect to $\{ i_0, i_1, \ldots, i_k \}$ and sum this inequality over $l = 0, \ldots, k$ to obtain 
    \begin{align}
        & \expect*{ \Phi(b^{k+1}) - \Phi(b^*) } + \expect*{ D_{1/\alpha^k}(b^*, b^{k+1}) } \nonumber \\ 
        \le& \left( \Phi(b^0) - \Phi(b^*) \right) - \frac{1}{n} \sum_{l=0}^k \expect*{ \Phi(b^l) - \Phi(b^*) } + D_{1/\alpha^0}(b^*, b^0) + \sum_{l=1}^k \expect*{ D_{1/\alpha^l}(b^*, b^l) - D_{1/\alpha^{l-1}}(b^*, b^l) }. 
        \label{eq:bcpr-convergence-inequality}
    \end{align}

    It is easy to see that when $\alpha_i^k = 1 / L_i$, the last term in \eqref{eq:bcpr-convergence-inequality} is zero and sublinear convergence rate follows after re-arrangement. 
    For line search case, we can derive \eqref{eq:bcpr-ls-bound}. 

    \paragraph{Remark. } Even though there are many equivalent formulas of the last term of \eqref{eq:bcpr-convergence-inequality} (see the following for examples), we cannot guarantee its convergence. 
    \begin{align}
        & \sum_{l=1}^k \expect*{ D_{1/\alpha^l}(b^*, b^l) - D_{1/\alpha^{l-1}}(b^*, b^l) } \nonumber \\ 
        =& \sum_{l=1}^k \expect*{ \sum_i \left( \frac{1}{\alpha_i^l} - \frac{1}{\alpha_i^{l-1}} \right) D_i\left( b_i^*, b_i^{l} \right) } \nonumber \\
        =& \sum_i \expect*{ \sum_{l=1}^k \left( \frac{1}{\alpha_i^l} - \frac{1}{\alpha_i^{l-1}} \right) D_i\left( b_i^*, b_i^{l} \right) } \nonumber \\
        =& \sum_i \expect*{ \sum_{l=1}^k \frac{1}{\alpha_i^{l-1}} \left( D_i(b_i^*, b_i^{l-1}) - D_i(b_i^*, b_i^l) \right) } - \expect*{ \frac{1}{\alpha_i^0} D_i(b_i^*, b_i^0) } + \expect*{ \frac{1}{\alpha_i^k} D_i(b_i^*, b_i^k) }. 
        \label{eq:diff-kl-convergence}
    \end{align}

\end{proof}

\subsection{Adaptive Block Coordinate descent algorithm for PR dynamics (\abcpr)}

We also give an adaptive stepsize strategy. 
To do this, we find a series of $\{ L_i^k \}_{k \in \NN}$ with each we have coordinate-wise relative smoothness with $L_i^k < 1$ at the $k$th iteration. 
We state the following lemma to support this algorithm. 

\begin{lemma}
    Let $b_{ij}^{k+1}$ be the next iterates of $b_{ij}^k$ in BCPR with $1 \le \alpha_i^k \le \bar{\alpha}_i^k,\;\forall\;i,k$, where $1$ and $\bar{\alpha}_i^k$ are lower and upper bounds of stepsizes, respectively. 
    Let 
    \begin{align}
        \beta_i^k = \left( {\max_{j: b_{ij}^* > 0} \frac{v_{ij}}{p_j^k}} \Big/ {\min_{j: b_{ij}^* > 0} \frac{v_{ij}}{p_j^k}} \right)^{\bar{\alpha}_i} 
        \label{eq:betai}
    \end{align} 
    and
    \begin{align} 
        \theta_i^k = \max_j \frac{b_{ij}^k}{p_j^k}. 
        \label{eq:thetai}
    \end{align}
    When $1 \le \beta_i^k \le \sqrt{2}$, we have 
    \begin{align}
        \varphi(b^{k+1}) \le \varphi(b^k) + \inp{\nabla_i \varphi(b^k)}{b_i^{k+1} - b_i^k} + L_i^k \DKL(b_i^{k+1}, b_i^k) 
        \label{eq:relative-smoothness-adaptive}
    \end{align}
    where 
    \begin{align}
        L_i^k = \frac{3}{4 - \beta_i^k}\left(\theta_i^k + \frac{2\beta_i^k - 1}{6\beta_i^k}{\theta_i^k}^2\right). 
        \label{eq:l-abcpr}
    \end{align} 
\label{lemma:a-bcpr}
\end{lemma}


\begin{proof}
    Let $c_j = p_j^+ - b_{ij}^+ = p_j - b_{ij}$. 
    The Taylor series of $\DKL(c_j + b_{ij}^+, c_j + b_{ij})$ with the Lagrange form of the remainder is  
    \begin{equation}
        \DKL(c_j + b_{ij}^+, c_j + b_{ij}) = \left( c_j + b_{ij} \right)\log{\frac{c_j + b_{ij}^+}{c_j + b_{ij}}} - b_{ij}^+ + b_{ij} = \frac{\left( b_{ij}^+ - b_{ij} \right)^2}{2\left( c_j + b_{ij} \right)} - \frac{\left( b_{ij}^+ - b_{ij} \right)^3}{6\left( c_j + b_{ij} \right)^2} + \frac{\left(b_{ij}^+ - b_{ij}\right)^4}{12\left(c_j + \lambda b_{ij}^+ + (1 - \lambda) b_{ij}\right)^3} \label{eq:taylor-series}
    \end{equation}
    for some $\lambda \in [0, 1]$. 

    We assume the unit budgets, i.e., $\sum_j b_{ij} = 1, \forall i$. 
    Since $1 \le \alpha_i \le \bar{\alpha}_i$, we can derive 
    \begin{equation}
        \frac{b_{ij}^+}{b_{ij}} = \frac{\left( \frac{v_{ij}}{p_j} \right)^{\alpha_i}}{\sum_l b_{il} \left( \frac{v_{il}}{p_l} \right)^{\alpha_i}} \le \frac{\max_j \left( \frac{v_{ij}}{p_j} \right)^{\alpha_i}}{\min_j \left( \frac{v_{ij}}{p_j} \right)^{\alpha_i}} = \left( \frac{\max_j \frac{v_{ij}}{p_j}}{\min_j \frac{v_{ij}}{p_j}} \right)^{\alpha_i} \le \left( \frac{\max_j \frac{v_{ij}}{p_j}}{\min_j \frac{v_{ij}}{p_j}} \right)^{\bar{\alpha}_i} = \beta_i. 
    \end{equation}

    Similarly, we have 
    \begin{equation}
        \frac{b_{ij}^+}{b_{ij}} = \frac{\left( \frac{v_{ij}}{p_j} \right)^{\alpha_i}}{\sum_l b_{il} \left( \frac{v_{il}}{p_l} \right)^{\alpha_i}} \ge \frac{\min_j \left( \frac{v_{ij}}{p_j} \right)^{\alpha_i}}{\max_j \left( \frac{v_{ij}}{p_j} \right)^{\alpha_i}} = \left( \frac{\min_j \frac{v_{ij}}{p_j}}{\max_j \frac{v_{ij}}{p_j}} \right)^{\alpha_i} \ge \left( \frac{\min_j \frac{v_{ij}}{p_j}}{\max_j \frac{v_{ij}}{p_j}} \right)^{\bar{\alpha}_i} = \frac{1}{\beta_i}. 
    \end{equation}

    Then, we obtain 
    \begin{align}
        \frac{\left( b_{ij}^+ - b_{ij} \right)^2}{2\left( c_j + b_{ij} \right)} - \frac{\left( b_{ij}^+ - b_{ij} \right)^3}{6\left( c_j + b_{ij} \right)^2} + \frac{\left(b_{ij}^+ - b_{ij}\right)^4}{12\left(c_j + \lambda b_{ij}^+ + (1 - \lambda) b_{ij}\right)^3} \ge& \frac{\left( b_{ij}^+ - b_{ij} \right)^2}{2\left( c_j + b_{ij} \right)} - \frac{\left( b_{ij}^+ - b_{ij} \right)^3}{6\left( c_j + b_{ij} \right)^2} \nonumber \\
        \ge& \left( \frac{1}{2(c_j + b_{ij})} - \frac{(\beta_i - 1)b_{ij}}{6(c_j + b_{ij})^2} \right) (b_{ij}^+ - b_{ij})^2. 
        \label{ineq:lower-bound}
    \end{align}

    On the other side, we have when $1 < \beta_i \le \sqrt{2}$, 
    \begin{align}
        \frac{\left( b_{ij}^+ - b_{ij} \right)^2}{2\left( c_j + b_{ij} \right)} - \frac{\left( b_{ij}^+ - b_{ij} \right)^3}{6\left( c_j + b_{ij} \right)^2} + \frac{\left(b_{ij}^+ - b_{ij}\right)^4}{12\left(c_j + \lambda b_{ij}^+ + (1 - \lambda) b_{ij}\right)^3} \le \left( \frac{1}{2(c_j + b_{ij})} + \frac{(1 - \frac{1}{2\beta_i}) b_{ij}}{6(c_j + b_{ij})^2} \right) (b_{ij}^+ - b_{ij})^2.  
        \label{ineq:upper-bound}
    \end{align} 

    To verify~\eqref{ineq:upper-bound}, we can re-arrange it and check this inequality is equivalent to  
    \begin{equation}
        \left(b_{ij}^+ - b_{ij}\right)^2 (c_j + b_{ij})^2 \le 2 \left(c_j + \lambda b_{ij}^+ + (1 - \lambda) b_{ij}\right)^3 \left( b_{ij}^+ - \frac{1}{2\beta_i} b_{ij} \right). 
    \end{equation}

    The above inequality can be verified by considering two cases. If $b_{ij}^+ \ge b_{ij}$, 
    \begin{align} 
        \left(b_{ij}^+ - b_{ij}\right)^2 (c_j + b_{ij})^2 &\le \left(b_{ij}^+ - b_{ij}\right)^2 \left(c_j + \lambda b_{ij}^+ + (1 - \lambda) b_{ij}\right)^2 \nonumber \\ 
        &= \left(b_{ij}^+ - b_{ij}\right) \left(c_j + \lambda b_{ij}^+ + (1 - \lambda) b_{ij}\right)^2 \left(b_{ij}^+ - b_{ij}\right) \nonumber \\
        &\le 2\left(c_j + \lambda b_{ij}^+ + (1 - \lambda) b_{ij}\right)^3 \left( b_{ij}^+ - \frac{1}{2\beta_i} b_{ij} \right)
    \end{align}
    where we use $b_{ij}^+ - b_{ij} \le 2 \left( c_j + b_{ij} + \lambda(b_{ij}^+ - b_{ij}) \right)$ and $b_{ij}^+ - b_{ij} \le b_{ij}^+ - \frac{1}{2\beta_i} b_{ij}$ when $1 < \beta_i \le 3$. 

    If $b_{ij}^+ < b_{ij}$, 
    \begin{align} 
        \left(b_{ij}^+ - b_{ij}\right)^2 (c_j + b_{ij})^2 \le b_{ij}^+ (c_j + b_{ij})^2 \left( b_{ij} - b_{ij}^+ \right) &\le \left(c_j + b_{ij}^+\right) \left(c_j + \beta_i b_{ij}^+\right)^2 \left( b_{ij} - b_{ij}^+ \right) \nonumber \\
        &\le \beta_i^2 \left(c_j + b_{ij}^+\right)^3 \frac{2}{\beta_i^2} \left( b_{ij}^+ - \frac{1}{2\beta_i} b_{ij} \right) \nonumber \\
        &\le 2\left(c_j + \lambda b_{ij}^+ + (1 - \lambda) b_{ij}\right)^3 \left( b_{ij}^+ - \frac{1}{2\beta_i} b_{ij} \right)
    \end{align}
    where we use $b_{ij} - b_{ij}^+ \le b_{ij}^+$ holds when $1 < \beta_i \le 2$, and 
    $
        b_{ij} - b_{ij}^+ \le \frac{2}{\beta_i^2} \left( b_{ij}^+ - \frac{1}{2\beta_i} b_{ij} \right)
    $
    holds when $1 < \beta_i \le \sqrt{2}$. 

    Combining \eqref{eq:taylor-series}, \eqref{ineq:lower-bound} and \eqref{ineq:upper-bound}, we attain 
    \begin{align}
        \frac{\varphi(b^+) - \varphi(b) - \inp{\nabla_i \varphi(b)}{b_i^+ - b_i}}{\DKL(b_i^+, b_i)} =& {\sum_j \DKL(c_j + b_{ij}^+, c_j + b_{ij})} \bigg/ {\sum_j \DKL(b_{ij}^+, b_{ij})} \nonumber \\
        \le& {\sum_j \left( \frac{1}{2(c_j + b_{ij})} + \frac{(1 - \frac{1}{2\beta_i}) b_{ij}}{6(c_j + b_{ij})^2} \right) (b_{ij}^+ - b_{ij})^2} \Bigg/ {\sum_j \left( \frac{1}{2 b_{ij}} - \frac{(\beta_i - 1)}{6 b_{ij}} \right) (b_{ij}^+ - b_{ij})^2} \nonumber \\
        \le& \max_j \left(\frac{1}{2(c_j + b_{ij})} + \frac{(1 - \frac{1}{2\beta_i}) b_{ij}}{6(c_j + b_{ij})^2}\right) \Bigg/ \left(\frac{1}{2 b_{ij}} - \frac{(\beta_i - 1)}{6 b_{ij}}\right) \nonumber \\ 
        =& \max_j \frac{3}{4 - \beta_i}\left( \frac{b_{ij}}{p_j} + \frac{2\beta_i - 1}{6\beta_i}\frac{b_{ij}^2}{p_j^2} \right) \label{ineq:final-bound}
    \end{align}
    where the second inequality follows since all terms are nonnegative in the two summations. 
    \cref{lemma:a-bcpr} follows because the right-hand side of~\eqref{ineq:final-bound} is increasing in $\frac{b_{ij}}{p_j}$. 
\end{proof}

Note that by the buyer optimality condition, when $b_{ij}^* > 0$, $v_{ij}/p_j^* = \max_j v_{ij}/p_j^*$. 
This implies $\beta_i^* = 1$. 
Hence, \cref{lemma:a-bcpr} is meaningful. 

We state adaptive block coordinate descent algorithm for PR dynamics as follows. 

\begin{algorithm}[H]
    \KwIn{$b^0 \in \mathbb{R}^{n\times m}, \bar{\alpha}_1, \ldots, \bar{\alpha}_n \in \mathbb{R}^+$}
        \For{$k \gets 1, 2, \ldots $} {
            pick $i_{k} \in [n]$ with probability ${1}/{n}$\;
            compute $L_i^k$ based on~\eqref{eq:betai}, \eqref{eq:thetai} and \eqref{eq:l-abcpr}\;
            $\alpha_{i_k} \gets \min\{\max\{1 / L_i^k, 1\}, \bar{\alpha_i}\}$\;
            compute $b_{i_{k}}^+$ based on~\eqref{eq:bcpr-update-simple} with stepsize $\alpha_{i_k}$\;
            $p^k \gets p^{k-1} - b_{ik} + b_{ik}^+$\;
            $b_{i_{k}}^{(k)} \gets b_{i_{k}}^{+} \quad \text{and} \quad b_{i'}^{(k)} \gets b_{i'}^{(k-1)}, \; \forall\, i' \neq i_{k}$\;
        }
\caption{Adaptive Block Coordinate Proportional Response (A-BCPR)}
\label{algo:a-bcpr}
\end{algorithm}

\section{Proportional Response with Line Search}
\label{app:prls}

We state general Mirror Descent with Line Search (MD-LS) here. 

\begin{algorithm}[H]
    \KwIn{Current iterate $b$, starting stepsize $\alpha$, number of LS steps in the previous iteration $p_{\rm prev}$.}
    \KwOut{current stepsize $\alpha_k$, next iterate $b^{k+1}$, current number of line search steps $p_k$.}

    \uIf{$p_{\rm prev} = 0$ (no backtracking in the previous iteration)} {
        set $\alpha^{(0)} = \min\{ \rho^+ \alpha, \alpha_{\max} \}$ (increment the stepsize)\;
    } 
    \Else{
        set $\alpha^{(0)} = \alpha$\;
    }
    \For{$p \gets 0, 1, 2, \ldots $} {
        \begin{enumerate}
            \item Compute $b^{(p)} = \argmin_{b'\in \mathcal{X}} \left\{ \langle \nabla \varphi(b), b' - b \rangle + \frac{1}{\alpha^{(p)}} D(b', b) \right\}$.
            \item Break if
            \[ \varphi(b^{(p)}) \leq \varphi(b) + \langle \nabla \varphi(b), b^{(p)} - b \rangle + \frac{1}{\alpha^{(p)}} D(b^{(p)}, b). \]
            \item Set $\alpha^{(p+1)} = \rho^- \alpha^{(p)}$ and continue.
        \end{enumerate}
    }
    Return $b^{k+1} = b^{(p)}$, $\alpha_k = \alpha^{(p)}$, $p_k = p$. 

    \caption{Line Search $(b^{k+1}, \alpha_k, p_k) \leftarrow \mathcal{LS}_{\rho^+, \rho^-, \alpha_{\max}}(b, \alpha, p_{\rm prev})$}
    \label{algo:ls-details}
\end{algorithm}

\begin{algorithm}[H]
    \KwIn{Initial iterate $b^0\in {\rm relint}(\mathcal{X})$, initial stepsize $\alpha^{-1} \leq \alpha_{\max}$, $p_{-1}=0$}
        \For{$k \gets 0, 1, 2, \ldots$} {
            Use \cref{algo:ls-details} to compute $(b^{k+1}, \alpha_k, p_k) = \mathcal{LS}_{\rho^+, \rho^-, \alpha_{\max}}(b^k, \alpha_{k-1}, p_{k-1})$. 
        }
    \caption{Mirror Descent with Line Search (MD-LS)}
    \label{algo:md-with-ls}
\end{algorithm}

\subsection{Proof of Theorem~\ref{thm:mdls-sublinear-conv}}

\begin{proof}

We prove sublinear convergence for general Mirror Descent with Line Search, which immediately implies \cref{thm:mdls-sublinear-conv} as we know $L=1$ for \eqref{program:shmyrev}. 

Same as \citet[Lemma 5]{birnbaum2011distributed}, it can be easily seen that \cref{algo:md-with-ls} is still a descent objective algorithm, that is, 
\begin{equation} 
    \varphi(b^{k+1}) \leq \varphi(b^k),\ k \in \NN. 
    \label{eq:mdls-objective-descent}
\end{equation}

By the convexity of $f$ and the break condition in \cref{algo:ls-details}, we know that
\begin{equation}
    \varphi(b^k) + \langle \nabla \varphi(b^k), b^{k+1} - b^k \rangle  \leq \varphi(b^{k+1}) \leq \varphi(b^k) + \langle \nabla \varphi(b^k), b^{k+1} - b^k \rangle + \frac{1}{\alpha^k} D(b^{k+1}, b^k). 
\end{equation}
Using \citet[Lemma 9]{birnbaum2011distributed} (originally in \citet{chen1993convergence}), we have, for any $u \in \mathcal{X}$, 
\begin{equation}
    \alpha^k \left( \varphi(b^k)+ \langle \nabla \varphi(b^k), b^{k+1}-b^k \rangle \right) + D(b^{k+1}, b^k) + D(u, b^k) \leq \alpha^k \left( \varphi(b^k) + \langle \nabla \varphi(b^k), u - b^k \rangle \right) + D(u, b^k)
\end{equation}
and then 
\begin{equation}
    \langle \nabla \varphi(b^k), b^{k+1} - b^k \rangle + \frac{1}{\alpha^k} D(b^{k+1}, b^k) \leq \langle \nabla \varphi(b^k), b - b^k \rangle + \frac{1}{\alpha^k} \left( D(u, b^k) - D(u, b^k) \right). 
    \label{eq:lemma-9-birnbaum}
\end{equation}

The above is similar to \citet[Corollary 10]{birnbaum2011distributed}. 
Let $b^*$ be an optimal solution. 
Then, similar to \citet[Lemma 8]{birnbaum2011distributed}, we have
\begin{align}
    \varphi(b^{k+1}) &\leq \varphi(b^k) + \langle \nabla \varphi(b^k), b^{k+1} - b^k \rangle + \frac{1}{\alpha^k} D(b^{k+1}, b^k) 
    \nonumber \\
    &\leq \varphi(b^k) + \langle \nabla \varphi(b^k), b^* - b^k \rangle + \frac{1}{\alpha^k} \left( D(b^*, b^k) - D(b^*, b^{k+1}) \right) 
    \nonumber \\
    &\leq \varphi(b^*) + \frac{1}{\alpha^k} \left( D(b^*, b^k) - D(b^*, b^{k+1}) \right). 
    \label{eq:lemma-8-birnbaum} 
\end{align}

This implies \cref{lem:mdls-d-descent} because 
\begin{equation}
    D(b^*, b^k) - D(b^*, b^{k+1}) \ge \alpha^k \left( \varphi(b^{k+1}) - \varphi(b^*) \right) \ge 0. 
\end{equation}

Summing up \eqref{eq:lemma-8-birnbaum} across $k = 0, 1, \ldots$, we have 
\begin{equation}
    k \left( \varphi(b^k) - \varphi^* \right) \le \sum_{l=0}^{k-1} ( \varphi(b^{l+1}) - \varphi^* ) \leq \sum_{l=0}^{k-1} \frac{D(b^*, b^l) - D(b^*, b^{l+1})}{\alpha^l}
    \label{eq:sum-up-lemma-8}
\end{equation}
where the first inequality follows from \eqref{eq:mdls-objective-descent}. Rearranging the above gives the following path-dependent bound: 
\begin{align}
    \varphi(b^k) - \varphi^* \leq \frac{1}{k} \left(\frac{1}{\alpha_0} D(b^*, b^0) - \frac{1}{\alpha_k} D(b^*, b^k) + \sum\nolimits_{l=1}^{k-1} \left(\frac{1}{\alpha_l} - \frac{1}{\alpha_{l-1}}\right) D(b^*, b^l)\right).
    \label{eq:mdls-tighter-convergence}
\end{align}
By \cref{lem:mdls-d-descent} and $\alpha^l \ge \frac{\rho_-}{L}$, \eqref{eq:mdls-non-tighter-convergence} follows from 
\begin{equation}
    k \left( \varphi(b^k) - \varphi^* \right) \le \sum_{l=0}^{k-1} \frac{D(b^*, b^l) - D(b^*, b^{l+1})}{\alpha^l} \le \frac{L}{\rho_-} \sum_{l=0}^{k-1} \left( D(b^*, b^l) - D(b^*, b^{l+1}) \right) \le \frac{L}{\rho_-} D(b^*, b^0). 
\end{equation} 

\end{proof}

\section{Block Coordinate Descent Algorithm for CES Utility Function}
\label{app:ces}

\subsection{Details of quadratic extrapolation for CES utility function}

In this section, we use $u_i^\rho$ to substitute $u_i(\rho)$ to avoid confusion as the decision variable of $u_i(\rho)$ is $x_i$. 

Denote the objective function of EG convex program for CES utility function as 
\begin{equation}
    f(x) = \sum_i f_i(x_i) = \sum_i g_i(u_i^\rho(x_i)) \quad \text{and} \quad u_i^\rho(x_i) = \sum_j u_{ij}^\rho(x_{ij})
\end{equation}
where 
\begin{equation}
    g_i(u_i^\rho) = - \frac{B_i}{\rho} \log{u_i^\rho} \quad \text{and} \quad u_{ij}^\rho(x_{ij}) = v_{ij} x_{ij}^\rho, \quad \rho \in (0, 1). 
\end{equation}

Then, we replace $g$ and $u^\rho$ with $\tilde{g}$ and $\tilde{u}^\rho$ (defined as follows), respectively. Similar to the linear case,  
\begin{align}
    \tilde{g}_i(u_i^\rho) := \begin{cases} 
        - \frac{B_i}{\rho} \log{u_i^\rho} & u_i^\rho \ge \underline{u}_i^\rho \\ 
        - \frac{B_i}{\rho} \log{\underline{u}_i^\rho} - \frac{B_i}{\rho \underline{u}_i^\rho} \left( u_i^\rho - \underline{u}_i^\rho \right) + \frac{B_i}{2 \rho {\underline{u}_i^\rho}^2} \left( u_i^\rho - \underline{u}_i^\rho \right)^2 & \text{otherwise} 
        \end{cases}, 
        \label{eq:ces-quad-extrapolation-g}
\end{align}
and 
\begin{align}
    \tilde{u}_{ij}^\rho(x_{ij}) := \begin{cases} 
        v_{ij} x_{ij}^\rho & x_{ij} \ge \underline{x}_{ij} \\ 
        v_{ij} \underline{x}_{ij}^\rho + \rho v_{ij} \underline{x}_{ij}^{\rho - 1} \left( x_{ij} - \underline{x}_{ij} \right) + \frac{1}{2}\rho(\rho - 1) v_{ij} \underline{x}_{ij}^{\rho - 2} \left( x_{ij} - \underline{x}_{ij} \right)^2 & \text{otherwise}  
        \end{cases} 
        \label{eq:ces-quad-extrapolation-u}
\end{align}
where $\underline{u}_i$ and $\underline{x}_{ij}$ are specific lower bounds for $u^*_i$ and $x^*_{ij}$ for some $i, j$. By similar discussion as \citet[Lemma 1]{gao2020first}, we know $\underline{u}_i$ exists for all $i$. By \cref{lem:ces-helper-lemma}, we know $\underline{x}_{ij}$ exists for all $i, j$. 

\subsection{Proof of \cref{lem:ces-helper-lemma}}

Here, we provide a proof for \cref{lem:ces-helper-lemma}, which leverages fixed point property at equilibrium like \citet[Lemma 8]{zhang2011proportional}, but generates a tighter bound. 
\begin{proof}
    Consider a buyer $i$ and an item $j$ with $v_{ij} > 0$. If for the buyer $i$, all other $k \neq j$ is associated with $v_{ik} = 0$, obviously $b_{ij}^* = B_i$ and $x_{ij}^* = b_{ij}^* / p_j^* \ge B_i / \sum_l B_l$. 
    
    Otherwise, if there is $v_{ik} > 0$ for some $k \neq j$, since a market equilibrium is a fixed point of the proportional response dynamics (and this can be extended to block coordinate case naturally), we have 
    \begin{equation}
        \left( \frac{v_{ij}}{p_j^*} \right)^\rho {b_{ij}^*}^{\rho - 1} = \left( \frac{v_{ik}}{p_k^*} \right)^\rho {b_{ik}^*}^{\rho - 1}. 
        \label{eq:fixed-point-eq}
    \end{equation}
    Rearrange \eqref{eq:fixed-point-eq} and then we get 
    \begin{equation}
        b_{ij}^* = \left( \frac{v_{ij}}{v_{ik}} \right)^{\frac{\rho}{1 - \rho}} b_{ik}^* \left( \frac{p_j^*}{p_k^*} \right)^{\frac{\rho}{\rho - 1}} 
    \end{equation}
    and 
    \begin{equation}
        x_{ij}^* = \frac{b_{ij}^*}{p_j^*} = \left( \frac{v_{ij}}{v_{ik}} \right)^{\frac{\rho}{1 - \rho}} b_{ik}^* {p_j^*}^{\frac{1}{\rho - 1}} {p_k^*}^\frac{\rho}{1 - \rho}. 
        \label{eq:ces-helper-lemma-x}
    \end{equation}

    To be specific, we consider $b_{ij}^* < \frac{B_i}{m}$ because otherwise the lower bound is directly given. 
    For these $i, j$, we can always choose another $k$ (there is at least one as discussed above) such that $b_{ik}^* \ge \frac{B_i}{m}$. 
    Hence, we can bound $p_j^*$ by $b_{ij}^* \le p_j^*$, \eqref{eq:fixed-point-eq} and $b_{ik}^* \le p_k^*$: 
    \begin{equation}
        p_j^* = {p_j^*}^{1 - \rho} {p_j^*}^{\rho} \ge {b_{ij}^*}^{1 - \rho} {p_j^*}^{\rho} = \left( \frac{v_{ij}}{v_{ik}} \right)^\rho {p_k^*}^\rho {b_{ik}^*}^{1 - \rho} \ge \left( \frac{v_{ij}}{v_{ik}} \right)^\rho {b_{ik}^*} \ge \omega_2(i)^\rho \frac{B_i}{m}
        \label{eq:p-opt-lower-bound}
    \end{equation}
    where $\omega_2(i) = \min_{j: v_{ij} > 0} v_{ij} / \max_j v_{ij}$. 
    Combining \eqref{eq:ces-helper-lemma-x} with \eqref{eq:p-opt-lower-bound}, we obtain 
    \begin{equation}
        x_{ij}^* = \left( \frac{v_{ij}}{v_{ik}} \right)^{\frac{\rho}{1 - \rho}} b_{ik}^* {p_j^*}^{\frac{1}{\rho - 1}} {p_k^*}^\frac{\rho}{1 - \rho} \ge \omega_2(i)^{\frac{\rho}{1-\rho}} \frac{B_i}{m} \left( \sum_l B_l \right)^{\frac{1}{\rho - 1}} \left( \omega_2(i)^\rho \frac{B_i}{m} \right)^{\frac{\rho}{1-\rho}} = \omega_1(i)^{\frac{1}{1-\rho}} \omega_2(i)^{\frac{\rho(1 + \rho)}{1 - \rho}}  
    \end{equation}
    where $\omega_1(i) = \frac{B_i}{m\sum_l B_l}$. All other cases give greater lower bounds than this. 
\end{proof}

\subsection{Proof of \cref{theorem:ces-convergence}}

We use the following lemma to establish a smoothness guarantee for $\tilde{f}(x) = \sum_i \tilde{g}(\sum_j \tilde{u}_{ij}^\rho (x_{ij}))$. 
\begin{lemma}
    For any $j \in [m]$, let 
    \begin{equation} 
        L_j = \max_{i\in [n]}{\frac{B_i v_{ij} \underline{x}_{ij}^{\rho - 2}}{\underline{u}_i(\rho)}}. 
    \label{eq:bcdeg-ces-stepsize} 
    \end{equation}
    Then, for all feasible $x, y$ such that $x, y$ differ only in the $j$th block, we have
    \begin{equation}
        \tilde{f}(y) \le \tilde{f}(x) + \inp{\nabla_{\cdot j} \tilde{f}(x)}{y - x} + \frac{L_{j}}{2}{{\lVert y - x \rVert}^2}. 
        \label{eq:bcdeg-ces-smoothness}
    \end{equation}
    \label{lemma:bcdeg-ces-smoothness}
\end{lemma}

\begin{proof}
Let $\tilde{u}_i$ stand for $\tilde{u}_i(\rho)$ for conciseness. We can obtain 
\begin{align}
    \nabla_{ij} \tilde{f} = \frac{\partial^2}{\partial x_{ij}^2} \tilde{f} = \frac{\partial^2 \tilde{g}_i}{\partial \tilde{u}_i^2} \left( \frac{\partial \tilde{u}_i}{\partial x_{ij}} \right)^2 + \frac{\partial \tilde{g}_i}{\partial \tilde{u}_i} \frac{\partial^2 \tilde{u}_i}{\partial x_{ij}^2} &\le \frac{B_i}{\rho \underline{u}_i^2} \left( \rho v_{ij} \underline{x}_{ij}^{\rho - 1} \right)^2 + \left( -\frac{B_i}{\rho \underline{u}_i} \right) \left( \rho (\rho - 1) v_{ij} \underline{x}_{ij}^{\rho - 2} \right) \nonumber \\ 
    &= \frac{B_i v_{ij} \underline{x}_{ij}^{\rho - 2}}{\underline{u}_i} \left( \rho \frac{v_{ij} \underline{x}_{ij}^\rho}{\underline{u}_i} + (1 - \rho) \right) \le \frac{B_i v_{ij} \underline{x}_{ij}^{\rho - 2}}{\underline{u}_i}. 
\end{align} 
Then, \eqref{eq:bcdeg-ces-smoothness} follows from the same derivation as \cref{lemma:coordinate-continuous}. 
\end{proof}

\cref{theorem:ces-convergence} can be obtained by \citet[Theorem 7]{richtarik2014iteration} and the fact that there exists a strong-convexity modulus $\mu(L)$ for $\tilde{f}$ w.r.t. the norm $\sum_j L_j {\lVert \cdot \rVert}_{\cdot j}^2$ when $\tilde{f}$ is of the form described above. 

\paragraph{Remark. }Similar to \cref{lemma:bcdeg-ces-smoothness}, we can establish a standard Lipschitz smoothness with constant $L = \max_{i \in [n]} \frac{\rho B_i \lVert v_i \rVert^2}{\underline{u}_i^2} \underline{x}_i^{2 \rho - 2}$. This can be seen a generalized result from \cref{lemma:full-continuous} as we have 
\begin{equation}
    \frac{\partial^2}{\partial x_{ij} \partial x_{ij'}} \tilde{f}(x) = \frac{\partial \tilde{u}_i}{\partial x_{ij}} \frac{\partial^2 \tilde{g}_i}{\partial \tilde{u}_i^2} \frac{\partial \tilde{u}_i}{\partial x_{ij'}} \le \rho v_{ij} v_{ij'} \frac{B_i}{\underline{u}_i^2} \underline{x}_{ij}^{\rho - 1} \underline{x}_{ij'}^{\rho - 1}. 
\end{equation}

\subsection{Proof of \cref{thm:bcpr-ces-convergence}}

We make use of an auxiliary function introduced in \citet{zhang2011proportional}, which they use to show convergence for deterministic PR.

\begin{proof}

In this proof, $D_i(\cdot, \cdot)$ denotes KL divergence of the $i$th block of the decision variables. Define 
\[
    T_{ij}^k = B_i \frac{v_{ij} \left( \frac{b_{ij}^k}{p_j^k} \right)^\rho}{\sum_{j'} v_{ij'} \left( \frac{b_{ij'}^k}{p_{j'}^k} \right)^\rho}, \quad u_{ij}^k(\rho) = v_{ij} \left( \frac{b_{ij}^k}{p_j^k} \right)^\rho, \quad u_i^k(\rho) = \sum_j v_{ij} \left( \frac{b_{ij}^k}{p_j^k} \right)^\rho, 
\]
thus $T_{ij}^k = B_i u_{ij}^k(\rho) / u_i^k(\rho)$. 

Let $m(k, i) = D_i(b_i^*, b_i^k)$ and $m(k) = \sum_i m(k, i)$, then we have 
\begin{equation} 
    \expect*{ m(k+1) | b^k } = \sum_i \frac{1}{n} \left\{ D_i (b_i^*, T_i^k) + \sum_{i' \neq i} m(k, i') \right\} = \frac{1}{n} \sum_i D_i(b_i^*, T_i^k) + \frac{n - 1}{n} m(k). 
\end{equation}

Similar to \citet{zhang2011proportional}, we obtain 
\begin{align}
    D_i(b_i^*, T_i^k) = \sum_j b_{ij}^* \log{\frac{b_{ij}^*}{T_{ij}^k}} = \sum_j b_{ij}^* \log{\frac{b_{ij}^* u_i^k(\rho)}{B_i u_{ij}^k(\rho)}} &= \sum_j b_{ij}^* \log{\frac{u_{ij}^*(\rho) u_i^k(\rho)}{u_i^*(\rho) u_{ij}^k(\rho)}} \nonumber \\
    &= \rho \sum_j b_{ij}^* \log{\frac{b_{ij}^*}{b_{ij}^k}} - \rho \sum_j b_{ij}^* \log{\frac{p_j^*}{p_j^k}} - \sum_j b_{ij}^* \log{\frac{u_i^*(\rho)}{u_i^k(\rho)}}. 
\end{align}

Hence, 
\begin{align}
    \expect*{ m(k+1) | b^k } &= \frac{1}{n} \sum_{i, j} \left\{ \rho b_{ij}^* \log{\frac{b_{ij}^*}{b_{ij}^k}} - \rho b_{ij}^* \log{\frac{p_j^*}{p_j^k}} - b_{ij}^* \log{\frac{u_i^*(\rho)}{u_i^k(\rho)}} \right\} + \frac{n-1}{n} m(k) \nonumber \\ 
    &= \frac{\rho}{n} m(k) + \frac{n-1}{n} m(k) - \frac{\rho}{n} \sum_j p_j^* \log{\frac{p_j^*}{p_j^k}} - \frac{1}{n} \sum_i B_i \log{\frac{u_i^*(\rho)}{u_i^k(\rho)}} \nonumber \\
    &= \left( 1 - \frac{1-\rho}{n} \right) m(k) - \frac{\rho}{n} \DKL(p^*, p^k) - \frac{1}{n} \sum_i B_i \log{\frac{\inp{v_i}{{x_i^*}^\rho}}{\inp{v_i}{{x_i^k}^\rho}}}. 
\end{align}

Since $\DKL(p^*, p^k) \ge 0$ and $\sum_i \frac{B_i}{\rho} \log{\frac{\inp{v_i}{{x_i^*}^\rho}}{\inp{v_i}{{x_i^k}^\rho}}} \ge 0$ (by the optimality of EG program), we have 
\begin{equation}
    \expect*{ m(k+1) | b^k } \le \rho' m(k), \quad \rho' = 1 - \frac{1 - \rho}{n}. 
\end{equation}
After taking expectation w.r.t. $i_1, i_2, \ldots, i_k$, we attain $\EE [m(k)] \le (\rho')^k m(0)$. 

Note that we have the following lemma. 
\begin{lemma}{\citep[Lemma 9]{zhang2011proportional}}
    For two positive sequences $\{ b^*_{ij} \}_{i \in [n], j \in [m]}$ and $\{ b_{ij} \}_{i \in [n], j \in [m]}$ such that $\sum_{i, j} b^*_{ij} = \sum_{i, j} b_{ij}$, let $\eta = \max_{i, j} \frac{\lvert b_{ij} - b^*_{ij} \rvert}{b^*_{ij}}$. 
    Then, 
\begin{equation}
    \DKL(b^*, b) \ge \frac{1}{8} \min\{1, \eta\} \eta \min_{i, j} b^*_{ij}. 
    \label{eq:zhang-dkl-eta}
\end{equation}
\end{lemma}

We consider the case where $\eta < 1$. From \citet[Lemma 8]{zhang2011proportional}, we know $b^*_{ij} \ge \left( \frac{1}{W^2} \right)^{\frac{1}{1-\rho}}$ for any $i, j$ such that $v_{ij} > 0$ where $W = \frac{n}{\min_{v_{ij} > 0} v_{ij} \cdot \min_i B_i}$. Then, \eqref{eq:zhang-dkl-eta} tells us that 
\begin{equation}
    m(k) = \DKL(b^*, b) \ge \frac{1}{8} \eta^2 \min_{i, j} b^*_{ij} \ge \frac{1}{8} \eta^2 \left( \frac{1}{W^2} \right)^{\frac{1}{1-\rho}} = \frac{1}{8} \left( \frac{\lvert b_{ij} - b^*_{ij} \rvert}{b^*_{ij}} \right)^2 \left( \frac{1}{W^2} \right)^{\frac{1}{1-\rho}}. 
\end{equation}

We take expectation on both sides of the above inequality, and use Jensen's inequality, then we have 
\begin{equation}
    \left( \expect*{ \frac{\lvert b_{ij} - b_{ij}^* \rvert}{b_{ij}^*} } \right)^2 \le \EE\left( \frac{\lvert b_{ij} - b^*_{ij} \rvert}{b^*_{ij}} \right)^2 \le 8 W^{\frac{2}{1-\rho}} \expect*{m(k)}. 
\end{equation}

Hence, to guarantee $\EE \frac{\lvert b_{ij} - b_{ij}^* \rvert}{b_{ij}^*} \le \epsilon < 1$, we need 
$8 W^{\frac{2}{1-\rho}} \expect*{m(k)} \le \epsilon^2$. 
Since $\EE [m(k)] \le (\rho')^k m(0)$, it suffices to have  
\begin{equation}
    (\rho')^k m(0) \le \frac{1}{8} \left( \frac{1}{W^2} \right)^{\frac{1}{1-\rho}} \epsilon^2, 
\end{equation}
which is equivalent to $k \ge { 2 \log{\frac{\sqrt{8 m(0)} W^{\frac{1}{1-\rho}}}{\epsilon}} } \Big/ {\log{\frac{n}{n - 1 + \rho}}}$. 

\end{proof}

\subsection{Discussion about BCD methods for CES utilities and more experimental results}

We showed that for linear utilities, the EG program can be solved efficiently by applying block coordinate descent algorithms. In contrast to linear utilities, there are some difficulties in computing market equilibria with CES ($\rho \in (0, 1)$) utilities by BCD-type methods. 

The key issue is that, while we can obtain a lower bound for $x_{ij}^*$ (for those $i, j: v_{ij} > 0$), this bound is too small to generate a reasonable (coordinate-wise) Lipschitz constant, since the lower bound appears in the the bound for $L_j$ (see \eqref{eq:bcdeg-ces-stepsize}).

Even though our algorithms are equipped with line search strategies, we find that they still struggle to find efficient stepsizes. 
We even find it hard to select an appropriate initial stepsize. 
Often, the stepsizes are too small to reach equilibrium, especially when the number of iterations is not very large. 
This makes \bcdegls\ and \pgls\ perform worse than \bcpr\ and \pr, because the latter do not need to choose stepsizes; they have (sublinear) convergence under stepsize $1$. 

This difficulty can be overcome when volatility of valuations is small - in this case, the lower bound of $x$ in \cref{eq:bcdeg-ces-stepsize} is larger. 

See \cref{fig:more-ces} for more experimental results on CES utilities. 
We used the model $v_{ij} = v_i v_j + \epsilon$ where $v_i \sim \mathcal{N}(1, \hat{v})$, $v_j \sim \mathcal{N}(1, \hat{v})$ and $\epsilon \sim \emph{uniform}(0, \hat{v})$, where $\hat{v}$ is a parameter to control volatility. Here, we set $n = m = 200$. 

\begin{figure}
    \centering
        \includegraphics[width=0.24\textwidth]{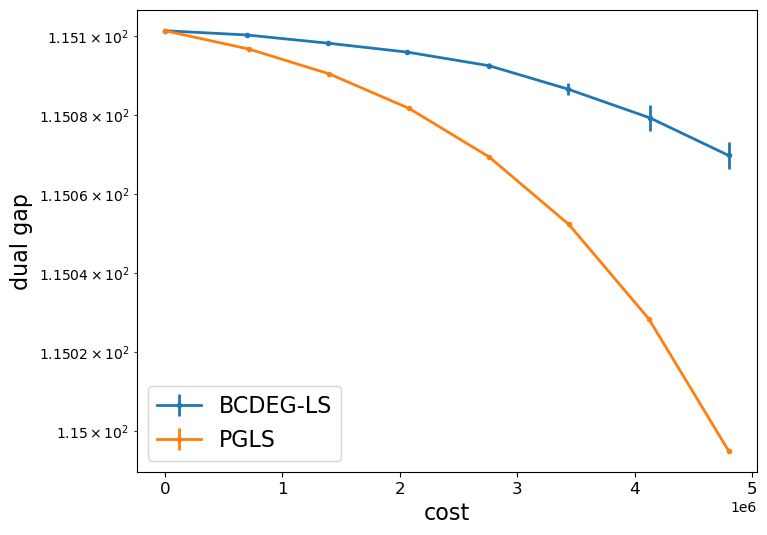}
        \includegraphics[width=0.22\textwidth]{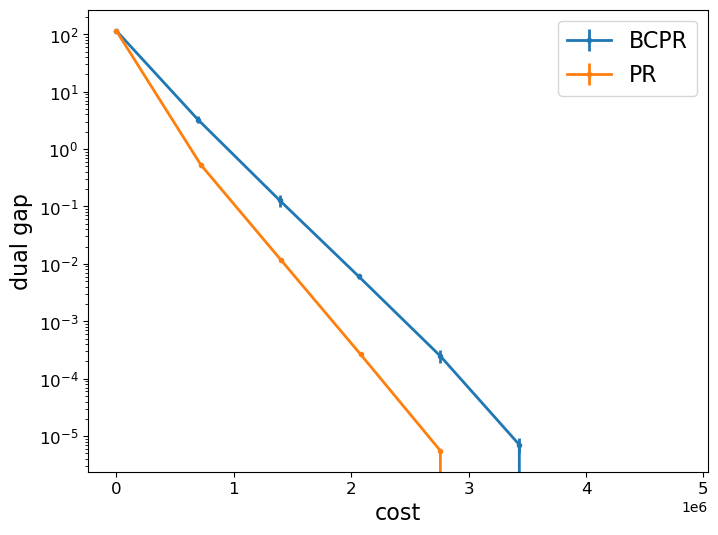}
        \includegraphics[width=0.24\textwidth]{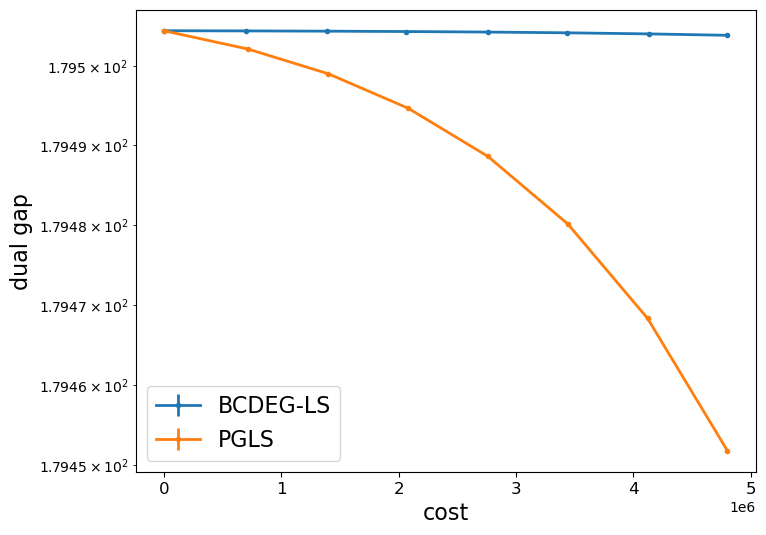}
        \includegraphics[width=0.22\textwidth]{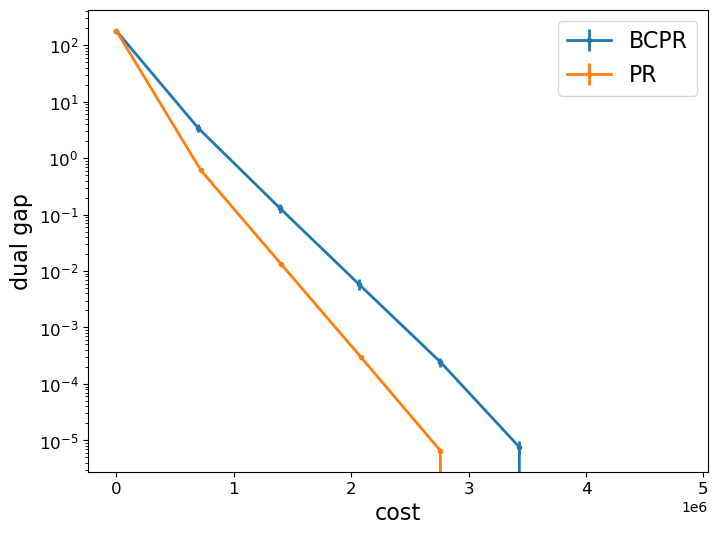}

        \includegraphics[width=0.24\textwidth]{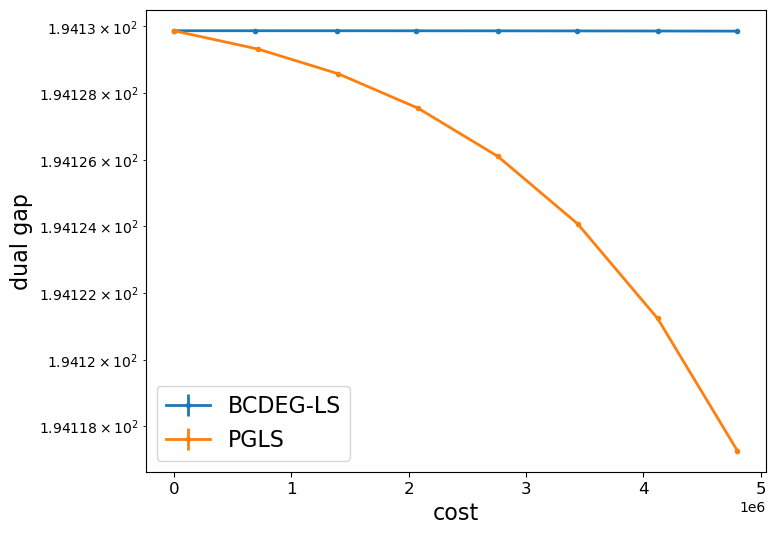}
        \includegraphics[width=0.22\textwidth]{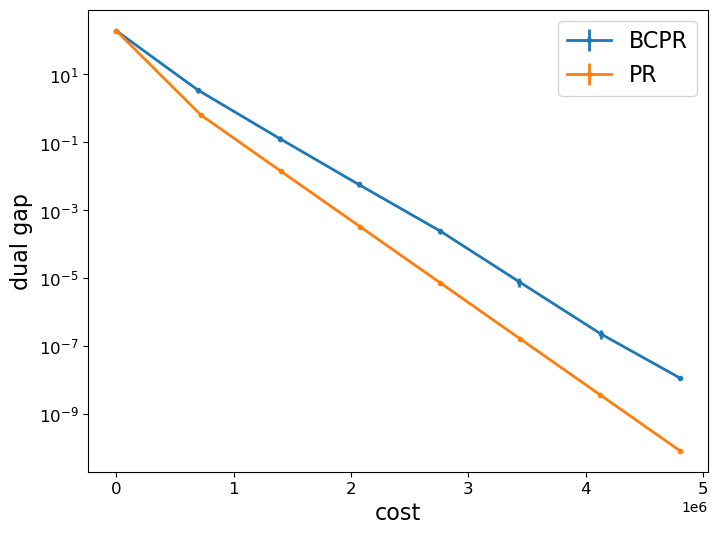}
        \includegraphics[width=0.24\textwidth]{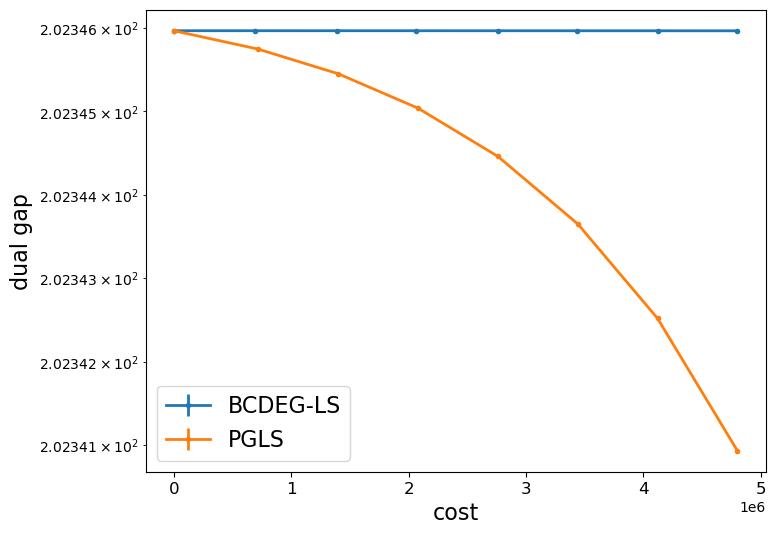}
        \includegraphics[width=0.22\textwidth]{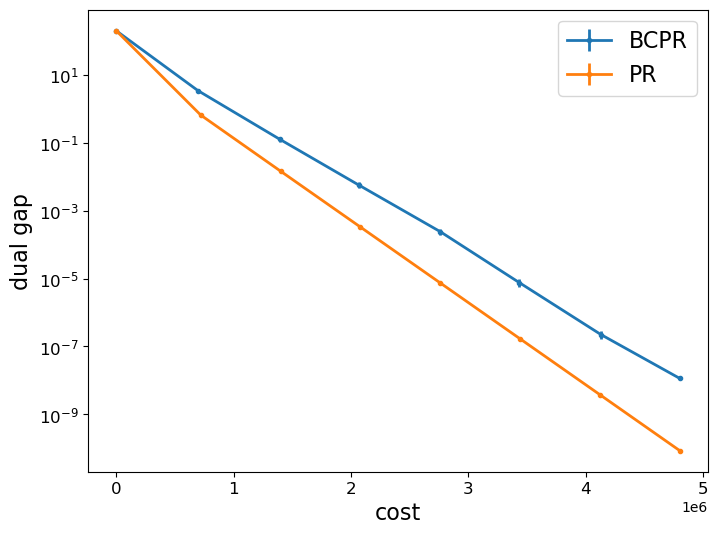}
        \vspace{-2mm}
    \caption{
        Performance on more simulated instances with CES utilities. ($\hat{v} = 0.4, 0.6, 0.8, 1.0$ from top to bottom). The plot setup is the same as in \cref{fig:compare-simulated-data}.
    }
    \label{fig:more-ces}
\end{figure}

\section{Discussion on wall-clock running times}

In this paper, we compare different algorithms based on the number of times they query for a buyer-item valuation (i.e. the number of accesses to cells of the valuation matrix). In this sense, if we count one query as one unit cost, then item-block algorithms (\bcdeg, \bcdegls) cost $n$ units in each iteration, buyer-block algorithms (\bcpr, \abcpr, \bcprls) cost $m$ units in each iteration, and full gradient algorithms (\pgls, \pr) cost $nm$ units in each iteration. Another way to compare these algorithms would be to compare wall-clock running times. 
In this section we justify our choice of measuring valuation accesses instead.


For the projected-gradient style methods (the full deterministic method and \bcdeg\ variants), they each perform one $n$-dimensional projection for every $n$ units of work, and thus all these methods are directly comparable in terms of our measure.
Similarly, all proportional-response-style algorithms perform the same amount of work for every $m$ queries made. Thus they are also directly comparable in terms of our measure.
When comparing projected-gradient-style methods to PR-style methods, there is a difference in that projection is required in the former. Yet in practice it is known that efficient implementations of the best projection-style algorithms tend to run in linear time as well
~\citep{condat2016fast}. 
Thus cross comparisons should also be fair.
The main reason why wall-clock might differ from our measure is due to efficiency of the implementation. 
All our experiments are implemented in python. For that reason, deterministic methods could have an advantage over block-coordinate methods in terms of wall clock, since NumPy and MKL enable highly-optimized matrix/vector operations. In turn, this means that deterministic methods can spend more time `in C' as opposed to `in python.'
If all algorithms are implemented without highly optimized matrix operations that exploit multi-core CPUs, experimental results on wall-clock running times will be consistent with our query-based complexity cost. 
Similarly, if they are all optimized fully to take advantage of hardware, parallelization, etc., we can still expect similar results. 




\end{document}